\documentclass[a4paper,11pt]{article}
\usepackage{latexsym}
\usepackage{amsmath, amsthm, amssymb, bbm}
\usepackage[mathscr]{eucal}
\usepackage{fancyhdr}
\usepackage{tikz}
\usepackage{cite}
\usepackage{hyperref}

\theoremstyle{plain}
\newtheorem{proposition}{Proposition}[section]

\newtheorem{lemma}[proposition]{Lemma}

\theoremstyle{definition}
\newtheorem{definition}[proposition]{Definition}

\theoremstyle{remark}
\newtheorem{remark}[proposition]{Remark}

\newcommand{\rhs}{r.h.s.\ }
\newcommand{\lhs}{l.h.s.\ }
\newcommand{\wrt}{w.r.t.\ }
\newcommand{\cf}{cf.\ }

\newcommand{\ud}{\mathrm{d}}
\newcommand{\del}{\partial}

\DeclareMathOperator{\supp}{supp}

\DeclareMathOperator{\Sym}{Sym}

\newcommand{\betrag}[1]{{\lvert #1 \rvert}}

\newcommand{\R}{\mathbb{R}}
\newcommand{\C}{\mathbb{C}}

\newcommand{\E}{\mathfrak{E}}
\newcommand{\F}{\mathfrak{F}}
\newcommand{\A}{\mathfrak{A}}

\newcommand{\D}{\mathfrak{D}}
\newcommand{\skal}[2]{\langle #1 , #2 \rangle}

\newcommand{\1}{\mathbbm{1}}

\newcommand{\diag}{\mathcal{D}}
\newcommand{\Spin}{\mathrm{Spin}}

\newcommand{\WDp}[1]{\colon \negthickspace #1 \! \colon \negthickspace }

\DeclareMathOperator{\WF}{WF}
\DeclareMathOperator{\Had}{Had}
\DeclareMathOperator{\tr}{tr}

\DeclareMathOperator{\Diag}{diag}

\newcommand{\ret}{{\mathrm{ret}}}
\newcommand{\adv}{{\mathrm{adv}}}
\newcommand{\reg}{{\mathrm{reg}}}
\newcommand{\loc}{{\mathrm{loc}}}

\newcommand{\ia}{{\mathrm{int}}}
\newcommand{\os}{{\mathrm{o.s.}}}

\newcommand{\CatVec}{\mathbf{Vec}}
\newcommand{\CatAlg}{\mathbf{Alg}}

\newcommand{\CatGSpMan}{\mathbf{GSpMan}}

\newcommand{\g}{\mathfrak{g}}
\newcommand{\T}{\mathfrak{T}}

\DeclareMathOperator{\ad}{ad}

\newcommand{\Lie}{\mathcal{L}}

\newcommand{\TO}{\mathcal{T}}
\newcommand{\Ret}{\mathcal{R}}
\newcommand{\Oo}{\mathcal{O}}

\begin{document}

\title{Locally covariant charged fields and background independence}
\author{Jochen Zahn \\  \\ Fakult\"at f\"ur Physik, Universit\"at Wien, \\ Boltzmanngasse 5, 1090 Wien, Austria. \\ jochen.zahn@univie.ac.at}

\date{\today}

\maketitle

\begin{abstract}
We discuss gauge background independence at the example of the charged Dirac field. We show that a perturbative version of background independence, termed perturbative agreement by Hollands and Wald, can be fulfilled, and discuss some of its consequences.
\end{abstract}

\section{Introduction}

The framework of locally covariant field theory \cite{BrunettiFredenhagenVerch} has proved very fruitful for quantum field theory (QFT) on curved spacetimes. The central idea is to define a quantum field theory simultaneously on all spacetimes, in a coherent way. Given a locally covariant field theory specified by a Lagrangean $L$, one may then wonder about background independence, i.e., is the field theory on a spacetime $M$ for the Lagrangean $L$ in some sense equivalent to a field theory defined on $M'$, if one adds $L_M - L_{M'}$ as an interaction term to the latter. This question was investigated for the scalar field in \cite{HollandsWaldStress}, and the requirement of background independence was formulated as the \emph{principle of perturbative agreement}. It takes the form of a renormalization condition, and it was shown that this principle can indeed be fulfilled for spacetime dimension $n>2$.

Recently, it was proposed to generalize the framework of locally covariant field theory to also accommodate for fields charged under a gauge group, in the presence of a background connection \cite{LocCovDirac}. In this setting, the background connection and gauge transformations are treated on equal footing with the background metric and its isometries. It is then natural to ask for background independence \wrt changes in the background connection.\footnote{I am very grateful to K.~Rejzner for proposing this question and helpful discussions on this topic.} The main result of the present work is that this background independence can indeed be fulfilled for the Dirac field for spacetime dimesion $n\leq 4$.\footnote{We do not expect obstructions to generalize this to arbitrary dimensions.} As a byproduct, we formulate the principle of perturbative agreement in the framework of locally covariant field theory, combined with that of \emph{perturbative algebraic QFT} \cite{BDF09}, i.e., in terms of functionals and natural transformations. Most of our results are in close analogy to those of \cite{HollandsWaldStress}, but we correct some minor mistakes in the proofs given there. We also present an application of perturbative agreement: We show that the fermionic contribution to the one-loop renormalization group flow of Yang--Mills theories can be obtained from the nontrivial scaling of Wick squares, i.e., of the parametrix. This is in the spirit of the background field method and shows the connection to heat kernel methods.

The article is structured as follows: In the next section, we introduce the setup. We review the definition of the charged locally covariant Dirac field given in \cite{LocCovDirac}. Subsection~\ref{sec:BackgroundVariation} deals with the question of how to relate functionals defined on different backgrounds, an issue that is crucial for our discussion. In Section~\ref{sec:PPA}, we first formulate the principle of perturbative agreement. Then, in Subsection~\ref{sec:CurrentConservation}, we show that if the principle of perturbative agreement for variations of the background connection is fulfilled, then the current is conserved, for semisimple gauge groups. We also show that current conservation can be ensured by a suitable choice of the parametrix. The background field method is discussed in Subsection~\ref{sec:BFM}. Finally, in Subsection~\ref{sec:FulfillmentPPA} we show that the principle of perturbative agreement can be fulfilled by a suitable redefinition of time-ordered products. In an appendix, we prove a technical lemma used in the proof of the fulfillment of perturbative agreement.

\section{Setup}

We review the definition of the charged locally covariant Dirac field given in \cite{LocCovDirac}. For further details, we refer to this reference and the ones cited therein, especially \cite{SandersDirac} for a detailed description of the definition of the spinor and Dirac bundles.

In the following, $G$ is a compact Lie group with Lie algebra $\g$. The symbol $\doteq$ stands for the definition of the left hand side by the right hand side. $\Gamma^\infty_{(c)}(M, B)$ denotes the space of (compactly supported) smooth section of the bundle $B$ over $M$. The causal future/past of a region $K$ is denoted by $J^\pm(K)$. The total diagonal of $M^k$ is denoted by $\diag_k$.

We use the following categories:
\begin{description}
\item[$\CatVec_{(i)}$:] The objects are locally convex vector spaces over $\C$. The morphisms are continuous linear (injective) maps.
%\item[$\CatIVec_{(i)}$:] The same as $\CatVec_{(i)}$, but restricted to involutive vector spaces and morphisms preserving the involution.
\item[$\CatAlg$:] The objects are topological unital $*$-algebras. The morphisms are continuous injective $*$-algebra homomorphisms.
\item[$\CatGSpMan$:] The objects are quintuples $(SM, P, \bar g, \bar A, \bar m)$, where $SM$ is a spin structure over $(M,\bar g)$, which is an oriented, time-oriented globally hyperbolic $n$-dimensional manifold. $P$ is a principal $G$ bundle over $M$, $\bar A$ a connection on $P$, and $\bar m \in C^\infty(M)$. A morphism $\chi: (SM, P, \bar g, \bar A, \bar m) \to (SM', P', \bar g', \bar A', \bar m')$ is given by $(\chi_{SM}, \chi_P)$, where $\chi_{SM (P)}$ is a principal $\Spin_0$ ($G$) bundle morphism. $\chi_{SM}$ and $\chi_P$ cover the same orientation, time-orientation and causality preserving isometric embedding $\psi: (M, \bar g) \to (M', \bar g')$, which is a diffeomorphism on its range. Furthermore, $\pi'_S \circ \chi_{SM} = \psi_* \circ \pi_S$, where $\pi_S$ is the spin projection from $SM$ to the (time-) oriented frame bundle $FM$, and $\bar A = \chi_P^* \bar A'$, $\bar m = \psi^* \bar m'$.
\end{description}

In the following, the background fields $(\bar g, \bar A, \bar m)$ are often subsumed in the symbol $\bar X$. 

We recall that the spin connection on $SM$ and the connection $\bar A$ on $P$ induce a unique connection on the principal $\Spin_0 \times G$ bundle $SM + P$ over $M$, which is obtained by taking the direct product bundle $SM \times P$ and restricting to the diagonal \cite{KobayashiNomizu}.

Given a finite-dimensional representation $\rho$ of $G$ on a $\C$ vector space $V$, we define the Dirac bundle
\[
 D_\rho M \doteq (SM + P) \times_{(\sigma, \rho)} (\C^{2^{[n/2]}} \otimes V),
\] 
where 
%$n$ is the spacetime dimension and 
$\sigma$ is the spinor representation. The dual bundle is denoted by $D_\rho^* M$ and the double Dirac bundle is the direct sum
\[
 D_\rho^\oplus M \doteq D_\rho M \oplus D_\rho^* M.
\]
Then we have a contravariant functor $\E^\oplus$ from $\CatGSpMan$ to $\CatVec$, by assigning to $(SM, P)$ the vector space of smooth sections,
\[
 \E^\oplus(SM, P) \doteq \Gamma^\infty(M, D_\rho^\oplus M),
\]
and to a morphism of $\CatGSpMan$ the pullback. Similarly, we define $\D^\oplus$ as a covariant functor $\CatGSpMan$ to $\CatVec_i$, by assigning to $(SM, P)$ the vector space of compactly supported smooth sections,
\[
 \D^\oplus(SM, P) \doteq \Gamma^\infty_c(M, D_\rho^\oplus M),
\]
and to a morphism of $\CatGSpMan$ the pushforward. Introducing a sesquilinear form on $V$ such that $\rho$ is unitary, one obtains a conjugation $+: V \to V^*$. Together with the Dirac conjugation, this induces conjugations $^+: D_\rho M \to D_\rho^* M$, $^+: D_\rho^* M \to D_\rho M$, which in turn yields a conjugation of $D_\rho^\oplus M$ by
\[
 (u,v)^* = (v^+, u^+),
\]
where $u \in D_\rho M|_x$, $v \in D_\rho^* M|_x$ for some $x \in M$. Using this pointwise definition, one defines the involution $^*$ on $\D^\oplus(SM, P)$ and  $\E^\oplus(SM, P)$.
%The vector spaces are equipped with their natural locally convex topology. The involution is defined as
%\begin{equation}
%\label{eq:Involution}
% (u,v)^* \doteq (v^+, u^+), 
%\end{equation}
%where $u \in \Gamma^\infty_{(c)}(M, D_\rho M)$, $v \in \Gamma^\infty_{(c)}(M, D_\rho^* M)$, and $^+$ denotes the Dirac conjugation.

On the vector bundle $D_\rho^\oplus M$, there is a bundle metric, induced by the pairing
\[
 \skal{[p, (y \otimes v, y' \otimes v')]}{[p, (z \otimes w, z' \otimes w')]} \doteq \skal{y'}{z} \skal{v'}{w} + \skal{z'}{y} \skal{w'}{v},
\]  
where $p \in SM+P$, $v, w \in V$, $y,z \in \C^{2^{[n/2]}}$ and the primed elements in the corresponding duals. As usual, the square brackets denote the equivalence class. This induces a pairing $\E^\oplus(SM, P) \times \E^\oplus(SM, P) \to C^\infty(M)$.
%Furthermore, there is an antilinear map $^+:\E^\oplus(SM, P) \to \E^\oplus(SM, P)$ mapping a section of the Dirac bundle to a section of the dual bundle and vice versa.

The connection on $SM + P$ induces a covariant derivative $\bar \nabla^\oplus = \bar \nabla \oplus \bar \nabla^*$ on the associated vector bundle $\E^\oplus(SM, P)$, which in turn induces the double Dirac operator
\[
 D^\oplus \doteq D \oplus - D^* \doteq (- \gamma^\mu \bar \nabla_\mu + \bar m) \oplus (- \gamma^\mu \bar \nabla^*_\mu - \bar m),
\]
where $\bar m$ is the smooth function specified in the object of $\CatGSpMan$. There is a corresponding causal propagator $S^\oplus = S^\oplus_\ret - S^\oplus_\adv$, where $S^\oplus_{\ret/\adv}$ is the retarded/advanced propagator. The double Dirac operator $D^\oplus$, and hence also $S^\oplus_{\ret/\adv}$, anticommutes with the involution.% \eqref{eq:Involution}.

In order to describe the coupling of Dirac fields to gauge fields, we also introduce the following vector bundles over $M$:
\begin{align*}
 E^0 & \doteq (SM + P) \times_{(\1, \ad)} \g, & E^k & \doteq E^0 \otimes \Omega^k(M).
\end{align*}
Here $\1$ denotes the trivial representation of $\Spin_0$ and in the definition of $E^k$, we take the tensor product of vector bundles. The representation $\rho$ and the spin projection induces an action of $E^0$ and $E^1$ on sections of the Dirac bundle (and hence also on the double Dirac bundle), fiberwisely given by
\begin{align*}
 ([p, \xi], [p, z \otimes v]) & \mapsto [p, z \otimes \rho(\xi) v], \\
 ([p, \xi] \otimes \omega, [p, z \otimes v]) & \mapsto [p, \omega_\mu \gamma^\mu z \otimes \rho(\xi) v].
\end{align*}
%%Here $\gamma_p$ is obtained as follows: $p \in SM + P$ determines a frame in $T_{\pi(p)} M$. The corresponding $n$ vectors are evaluated in the covector $\omega$, and the resulting number is a coefficient for the action of the corresponding Dirac matrix.
%There is also an action of sections of $E^0$ on sections of $E^0$ and $E^1$, fiberwisely defined by
%\begin{align*}
% [p, \xi] \wedge [p, \eta] & \doteq [p, [\xi, \eta]], & [p, \xi] \wedge [p, \eta \otimes \omega] & \doteq [p, [\xi, \eta] \otimes \omega].
%\end{align*}
%Finally, there is map $\bar \ud: \Gamma^\infty(M, E^0) \to \Gamma^\infty(M, E^1)$ given by
%\[
% (\bar \ud c)(V) \doteq \bar \nabla_V c,
%\]
%where $V \in TM$ and $\bar \nabla$ is the covariant derivative induced by the background connection on $E^0$.
There is also a product $\wedge: \Gamma^\infty(M, E^k) \times \Gamma^\infty(M, E^l) \to \Gamma^\infty(M, E^{k+l})$, fiberwisely defined by
\begin{multline*}
 [p, \xi] \otimes \omega \wedge [p, \eta] \otimes \nu \doteq [p, [\xi, \eta]] \otimes \omega \wedge \nu, \\
 p \in SM+P, \xi, \eta \in \g, \omega \in \Omega_{\pi(p)}^k, \nu \in \Omega_{\pi(p)}^l.
\end{multline*}
Finally, there is map $\bar \ud: \Gamma^\infty(M, E^k) \to \Gamma^\infty(M, E^{k+1})$ given by
\[
 \bar \ud ( a \Xi ) = \bar \nabla_\mu a \ud x^\mu \wedge \Xi + a \ud \Xi, \qquad a \in \Gamma^\infty(M, E^0), \Xi \in \Gamma^\infty(M, \Omega^k),
\]
where $\bar \nabla$ is the covariant derivative induced by the background connection on $E^0$, \cf \cite{MaratheMartucci89} for more details on the construction.

The vector space of test tensors is now defined as (this is a generalization of the definition used in \cite{LocCovDirac})
\begin{equation}
\label{eq:Tens}
 \T_c(SM, P) \doteq \Gamma^\infty_c(M, \bigotimes D_\rho^\oplus M \otimes \bigotimes T M \otimes \bigotimes E^0 \otimes \bigotimes E^1),
\end{equation}
where $\bigotimes$ denotes the direct sum of the tensor product bundles of all orders.
%As on $\E^\oplus$, the involution is defined through the pointwise Dirac involution of $D_\rho^\oplus M$, supplemented by a reversion of the order in all tensor products, i.e., in each of the factors $\bigotimes D_\rho^\oplus M, \dots, \bigotimes E^1$.
This defines a covariant functor from $\CatGSpMan$ to $\CatVec_i$.

%\begin{remark}
%\label{rem:density}
%In order to have a good behavior under variations of the metric, it would be advantageous to use densitized test tensors, as in \cite{HollandsWaldStress}, as discussed below. The present definition has the advantage that the functional differentiation \wrt background fields gives fields, \cf Definition~\ref{def:FunctionalDiffField} and  Definition~\ref{def:Current}. As we are primarily interested in perturbations of the background connection, which do not affect the volume element, we prefer the present definition.
%\end{remark}

%\begin{equation}
%\label{eq:Tens}
% \T_c(SM, P) \doteq \Gamma^\infty_c(M, \wedge( D_\rho^\oplus M \otimes T^\oplus M) \otimes \Sym E^0 \otimes \Sym E^1),
%\end{equation}
%where
%\[
% T^\oplus M \doteq \bigoplus_{k} \Sym^k TM,
%\]
%$\wedge$ denotes the antisymmetric tensor product of vector bundles, and $\Sym^k$ the $k$th symmetric tensor product.

\subsection{Functionals}

In the functional approach \cite{DuetschFredenhagenDeformation, BDF09, RejznerFermions}, one considers the algebra of functionals on a vector space of configurations, in the present case
\[
 \wedge \E^\oplus(SM, P) \doteq  \bigoplus_{k = 0}^\infty \wedge^k \E^\oplus(SM, P),
\]
with
\[
 \wedge^k \E^\oplus(SM, P) \doteq \{ B \in \Gamma^\infty(M^k, (D_\rho^\oplus M)^k) | B \text{ antisymmetric} \}.
\]
This space is equipped with its natural topology (uniform convergence of all derivatives on compact subsets). For an element $B \in \wedge \E^\oplus(SM, P)$, we denote by $B_k$ its component in $\wedge^k \E^\oplus(SM, P)$. On $\wedge \E^\oplus(SM, P)$, the conjugation $^*$ acts by fiberwise conjugation and reversal of the order of the arguments.

We now consider functionals on $\wedge \E^\oplus(SM, P)$, i.e., linear maps from this space into the complex numbers. The restriction of a functional $F$ to $\wedge^k \E^\oplus(SM, P)$ is denoted by $F_k$, and the grade $\betrag{F_k}$ of $F_k$ is $k$. The \emph{regular} functionals, $\F_\reg(SM, P)$, are those of the form
\begin{equation}
\label{eq:F_reg}
 F_k(B) = \int \skal{f_k}{B_k}(x_1, \dots, x_k) \ud_{\bar g} x_1 \dots \ud_{\bar g} x_k,
\end{equation}
with $f_k \in \Gamma^\infty_c(M^k, D_\rho^\oplus M^k)$, $f_k$ antisymmetric. Here $\ud_{\bar g} x$ is the volume element corresponding to the background metric. $f_k$ is called the \emph{kernel} of $F_k$. Here we used the obvious generalization of the pairing of sections of the double Dirac bundle. The \emph{support} of a functional is defined as the support of its kernel,
\[
 \supp F_k \doteq \supp_M f_k \doteq \{ x \in M | (x, x_2, \dots, x_k) \in \supp f_k \text{ for some } x_i \}.
\]

On $\F_\reg(SM, P)$, one introduces an antisymmetric product $\wedge$, by defining the kernel of the product $F \wedge H$ as
\begin{multline*}
 (f \wedge h)_k(x_1, \dots, x_k) \\ \doteq \frac{1}{k!} \sum_{l=0}^k  \sum_{\pi \in S_k} (-1)^{\betrag{\pi}} f_l(x_{\pi(1)}, \dots, x_{\pi(l)}) h_{k-l}(x_{\pi(l+1)}, \dots, x_{\pi(k)}).
\end{multline*}
An involution on $\F_\reg(SM, P)$ is defined as
\[
 F^*(B) \doteq \overline{F(B^*)}.
\]
Finally, we equip $\F_\reg(SM, P)$ with the topology induced from the standard locally convex topology on $\Gamma^\infty_c(M^k, D_\rho^\oplus M^k)$ (uniform convergence of all derivatives on compact sets), the space of the kernels.
The assignment $(SM, P) \mapsto \F_\reg(SM, P)$ is then a covariant functor from $\CatGSpMan$ to $\CatAlg$, where
\begin{equation}
\label{eq:F_Morphism}
 \F_\reg(\chi) (F)(B) \doteq F(\chi^* B)
\end{equation}
for a morphism $\chi$.

The regular functionals do not allow for the description of local interactions or nonlinear observables, such as the current. In order to cure this, one allows for more general kernels $f_k$, namely compactly supported distributions fulfilling the wave front set condition
\begin{equation}
\label{eq:microcausal}
 \WF(f_k) \cap (\bar V_+^k \cup \bar V_-^k) = \emptyset,
\end{equation}
where $\bar V_{\pm}$ is the closure of the dual of the forward/backward light cone. These are called the \emph{microcausal} functionals. They also form an algebra $\F(SM, P)$. It can be equipped with a topology such that it is a nuclear, locally convex vector space \cite{BDF09, Rej11}, see also \cite{DabrowskiBrouder} for a detailed discussion of this topology. $\F$ is then also a covariant functor from $\CatGSpMan$ to $\CatAlg$, with the action on morphisms as in \eqref{eq:F_Morphism}.

The subspace $\F_\loc(SM, P)$ of $\F(SM, P)$ in which the $f_k$'s are localized on the total diagonal $\diag_k$
with $\WF(f_k) \perp T \diag_k$ is the space of \emph{local} functionals. It is a covariant functor from $\CatGSpMan$ to $\CatVec_i$. We note that by \cite[Thm.~2.3.5]{HoermanderI}, the kernels are then of the form
\begin{equation}
\label{eq:LocalKernel}
 f_k(x_1, \dots, x_k) = \int f(x) \delta^{\alpha_1}(x,x_1) \dots \delta^{\alpha_k}(x,x_k) \ud_{\bar g} x,
\end{equation}
where $f$ is a compactly supported smooth section and the $\alpha_i$ are multiindices.

We denote by $\F_0(SM, P, \bar X)$ the ideal of functionals that vanish on all on-shell configurations, i.e., on configurations fulfilling $D^\oplus B = 0$, where $D^\oplus$ acts on an arbitrary coordinate. We define the on-shell functionals as $\F^\os(SM, P, \bar X) \doteq \F(SM, P) / \F_0(SM, P, \bar X)$. This amounts to identifying two functionals if they agree on all on-shell configurations. This is also a covariant functor from $\CatGSpMan$ to $\CatAlg$.

\begin{remark}
\label{rem:OnShell}
$F \in \F_0(SM, P, \bar X)$ implies that $f_k = \sum_j D^\oplus_j g^j_k$, where the $g_k^j$ are compactly supported distributional sections on $M^k$ fulfilling the wave front set condition \eqref{eq:microcausal} and $D^\oplus_j$ acts on the $j$th coordinate. To see this for $k=1$, choose a compactly supported smooth function $\chi$ such that $\chi=1$ on $\supp f_1$. Given $B_1$, define $\tilde B_1 \doteq S^\oplus_\ret(\chi D^\oplus B_1)$. Then $D^\oplus(B_1 - \tilde B_1)|_{\supp f_1} = 0$, so that $F(B) = F(\tilde B)$. It follows that $f_1$ can be written as $f_1 = D^\oplus (\chi S^\oplus_\adv(f_1))$. This straightforwardly generalizes to $k>1$.
\end{remark}

Functional derivatives are defined as follows \cite{RejznerFermions}:
\[
 F^{(1)}(B)(u) \doteq F(u \wedge B), \quad B \in \wedge \E^\oplus(SM, P), u \in \E^\oplus(SM, P).
\]
Hence, $F^{(1)}(B)$ can be interpreted as a compactly supported distributional section of $D_\rho^\oplus M$. We denote its integral kernel by $F^{(1)}(B)(x)$. For $F \in \F_\reg$, this is even a smooth section. The functional $B \mapsto F^{(1)}(B)(u)$ will in the following be denoted by $F^{(1)}(u)$.

\subsection{Quantization}

Quantization in the sense of deformation quantization \cite{DuetschFredenhagenDeformation} is straightforward for regular functionals, by defining the $\star$ product via functional derivatives and convolution with $S^\oplus$. In order to proceed to microcausal functionals, one uses Hadamard two-point functions:

\begin{definition}
A \emph{Hadamard two-point function} is a distributional section $\omega \in \Gamma^\infty_c(M^2, D_\rho^\oplus M^2)'$ fulfilling
\begin{align}
\label{eq:HadamardWaveEq}
 \omega(D^\oplus u, v) & = 0, \\
\label{eq:HadamardAnticommutator}
 \omega(u,v) + \omega(v,u) & = i S^\oplus(u,v), \\
\label{eq:HadamardConjugation}
 \overline{\omega(u, v)} & = \omega(v^*,u^*), \\
\label{eq:HadamardWF}
 \WF(\omega) & \subset C_+,
\end{align}
where $u, v \in \Gamma^\infty_c(M, D_\rho^\oplus M)$ and
\[
 C_\pm = \{ (x_1, x_2; k_1, - k_2) \in T^* M^2 \setminus \{ 0 \} | (x_1; k_1) \sim (x_2; k_2), k_1 \in \bar V^\pm_{x_1} \}.
\]
Here $(x_1; k_1) \sim (x_2; k_2)$ if there is a lightlike geodesic joining $x_1$ and $x_2$ to which $k_1$ and $k_2$ are co-parallel and co-tangent. For $x_1 = x_2$, $k_1, k_2$ are lightlike and coinciding.
\end{definition}

Note that \eqref{eq:HadamardWaveEq} and \eqref{eq:HadamardAnticommutator} imply that $\omega$ is in fact a bi-solution.

The existence of Hadamard two-point functions for arbitrary backgrounds was proven in \cite{LocCovDirac}.
For a Hadamard two-point function $\omega$, one defines
\begin{equation}
\label{eq:star_omega}
 F \star_\omega G \doteq \wedge \exp(\hbar \Gamma^\otimes_{\omega}) F \otimes G,
\end{equation}
where $F, G \in \F(SM, P)$, the wedge denotes the wedge product, $\wedge (F \otimes G) \doteq F \wedge G$, and
\[
 \Gamma^\otimes_{\omega} (F \otimes G) \doteq (-1)^{\betrag{F}+1} \int F^{(1)}(x) \otimes G^{(1)}(y) \omega(x,y) \ud_{\bar g} x \ud_{\bar g} y.
\]
Due to Remark~\ref{rem:OnShell} and the fact that $\omega$ is a bi-solution, this is well-defined also on $\F^\os(SM, P, \bar X)$.
%As $\omega$ is a bi-solution this is well-defined also on $\F_S(SM, P, \bar X)$, by \eqref{eq:HadamardWaveEq} and \eqref{eq:HadamardAnticommutator}.
In order to make covariance explicit, consider the set $\Had(SM, P, \bar X)$ of all Hadamard two-point functions on the background $(SM, P, \bar X)$, and define the space $\A(SM, P, \bar X)$ of \emph{quantum functionals} as the space of families
\begin{align*}
 F & = \{ F_\omega \}_{\omega \in \Had(SM, P, \bar X)}, & F_\omega \in \F(SM, P)[[\hbar]]
\end{align*}
fulfilling
\begin{equation}
\label{eq:F_omega'}
 F_{\omega'} = \exp(\hbar \Gamma_{\omega'-\omega}) F_\omega,
\end{equation}
where
\begin{equation}
\label{eq:DefGamma}
 \Gamma_{\omega} F \doteq \int \omega(x,y) F^{(2)}(x,y) \ud_{\bar g} x \ud_{\bar g} y.
\end{equation}
An element $F$ of $\A(SM, P, \bar X)$ is entirely specified by \eqref{eq:F_omega'} and $F_\omega$ for a single $\omega \in \Had(SM, P, \bar X)$.
We equip $\A(SM, P, \bar X)$ with the product
\[
 (F \star G)_\omega = F_\omega \star_\omega G_\omega,
\]
and the involution
\[
 (F^*)_\omega = (F_\omega)^*.
\]
The condition \eqref{eq:HadamardConjugation} ensures that this is consistent.

The assignment $(SM, P, \bar X) \mapsto (\A(SM, P, \bar X), \star)$ is a covariant functor from $\CatGSpMan$ to $\CatAlg$,
with
\begin{equation*}
%\label{eq:A_morphism}
 (\A(\chi) F)_{\omega'} \doteq \F[[\hbar]](\chi)(F_{\chi^* \omega'}),
\end{equation*}
with $\F(\chi)$ defined by \eqref{eq:F_Morphism}.
%which maps a morphism $\chi$ to the morphism $\chi_*$ defined by
%\begin{equation}
%\label{eq:A_morphism}
% (\chi_* F)_\omega = \chi_*(F_{\omega|_{M \times M}}),
%\end{equation}
%where on the \rhs $\chi_*$ is the morphism of $\F[[\hbar]]$.
We define the algebra $\A^\os(SM, P, \bar X)$ of on-shell functionals analogously to $\F^\os(SM, P, \bar X)$. The local elements $\A_\loc(SM, P, \bar X)$ of $\A(SM, P, \bar X)$ are defined as those for which $F_\omega \in \F_\loc(SM, P)[[\hbar]]$ for one (and hence all) $\omega$. Again, $\A_\loc$ is a covariant functor from $\CatGSpMan$ to $\CatVec_i$.

By evaluating a state on $\A^\os(SM, P, \bar X)$ on products of linear functionals, i.e., functionals of the form \eqref{eq:F_reg} with $k=1$, multiplied with the $\star$ product, one obtains the $n$-point functions of the state. The truncated $n$-point functions are defined as usual.

\begin{definition}
A \emph{Hadamard state} is a state whose two-point function is a Hadamard two-point function and whose truncated $n$-point functions are smooth.
\end{definition}

The importance of Hadamard states is that the set of all Hadamard states is identical to the set of continuous functionals on $\A^\os(SM, P, \bar X)$ \cite{HollandsRuan}.

\subsection{Fields}

In the framework of locally covariant field theory \cite{BrunettiFredenhagenVerch}, fields allow to define functionals on different backgrounds in a coherent way. A \emph{field} $\Phi$ is a natural transformation from $\T_c$ to $\F_\loc$, i.e., to each background $(SM, P, \bar X)$ it associates a linear map $\Phi_{(SM, P, \bar X)}: \T_c(SM, P) \to \F_\loc(SM, P)$, such that, for a morphism $\chi:(SM,P,\bar X) \to (SM',P',\bar X')$,
\begin{equation}
\label{eq:Field}
 \F(\chi) \Phi_{(SM, P, \bar X)}(t) = \Phi_{(SM', P', \bar X')}(\T_c(\chi) t).
\end{equation}
Leaving the test tensor open, we can then interpret the kernel of $\Phi_{(SM, P, \bar X)}(\cdot)_k$ as a distributional section on $M^{k+1}$. We require that it is supported on the diagonal $\diag_{k+1}$, has finite order, and fulfills
\[
 \WF(\Phi_{(SM, P, \bar X)}(\cdot)_k) \perp T \diag_{k+1}.
\]
Hence, it can be written in a form analogous to \eqref{eq:LocalKernel}, but with $f$ not compactly supported.\footnote{Note that here one has to generalize \cite[Thm.~2.3.5]{HoermanderI} to non-compactly supported distributions, which is possible.}
A further condition implying a smooth and (if applicable) analytic dependence on the background data will be given in Section~\ref{sec:BackgroundVariation} below.

A $k$-\emph{local field} is a natural transformation from the functor $\T_c^{\otimes k}$ to the functor $\F$ (interpreted as a functor from $\CatGSpMan$ to $\CatVec_i$), preserving the support. $k$-local fields for $k>1$ are also called \emph{multilocal}. A \emph{quantum field} is a natural transformation from the functor $\T_c$ to the functor $\A_\loc$, which preserves the support, and analogously for multilocal quantum fields. Similarly, \emph{on-shell (quantum) (multilocal) fields} are natural transformations to the corresponding on-shell functors.

%\begin{definition}
%A \emph{field} is a natural transformation from the functor $\T_c$ to the functor $\F_\loc$, which preserves the support. A $k$-\emph{local field} is a natural transformation from the functor $\T_c^{\otimes k}$ to the functor $\F$ (interpreted as a functor from $\CatGSpMan$ to $\CatVec_i$), also preserving the support. $k$-local fields for $k>1$ are also called \emph{multilocal}. A \emph{quantum field} is a natural transformation from the functor $\T_c$ to the functor $\A_\loc$, which preserves the support, and analogously for multilocal quantum fields.
%\end{definition}

Examples for fields are the monomials that map a test section $t_k \in \Gamma^\infty_c(M, \wedge^k (D_\rho^\oplus M \otimes \bigotimes T M))$ to the functional
\[
 \Phi_{(SM, P, \bar X)}(t_k)(B) = \int \skal{t_k^{\underline{\mu_1} \dots \underline{\mu_k}}}{\bar \nabla^{\oplus 1}_{\underline{\mu_1}} \dots \bar \nabla^{\oplus k}_{\underline{\mu_k}} B_k}(x, \dots, x) \ud_{\bar g} x,
\]
where $\underline{\mu_i}$ is the multiindex corresponding to the $\bigotimes TM$ part of the $i$th factor of $\wedge^k (D_\rho^\oplus M \otimes \bigotimes T M)$, 
%where the $j$ tensor indices are grouped into $k$ multiindices
and $\bar \nabla^{\oplus i}$ denotes the covariant derivative \wrt the $i$th coordinate. One can thus generate all local functionals by fields. Hence, when aiming at proving a statement for local functionals, one can equivently prove it for fields.

An example for a field that we will encounter in the following is, for $A \in \Gamma^\infty_c(M, E^1)$,
\begin{equation}
\label{eq:DefCurrent}
 j_{(SM, P, \bar X)}(A)(B) \doteq \int (\gamma^\mu)^\alpha_\beta \rho(A_\mu(x))^a_b (B_2)_{a\alpha}^{\ \ \ b\beta}(x,x) \ud_{\bar g}x.
\end{equation}
Here $\alpha, \beta$ are double spinor indices, $a, b$ are gauge indices, and we pick the component of $B_2$, whose first entry is in the dual Dirac bundle and the second in the Dirac bundle.
%where the first coordinates takes values in the dual bundle and the second coordinate in the bundle.
This is nothing but the Lie algebra valued current of Yang-Mills theories. In a succinct notation, we may write it as
\begin{equation}
\label{eq:CurrentSuccinct}
 j^\mu_I = \psi^+ \gamma^\mu T_I \psi,
\end{equation}
where $I$ is a Lie algebra index and $T_I$ the corresponding generator in the representation $\rho$.
%, and it can be obtained from the Dirac Lagrangean by functional differentiation, as follows. The space of connections of a principal bundle is an affine space. Hence, after the choice of a (background) connection, all other connections can be obtained by adding an element of $\Gamma^\infty(M, E^1)$.
As we will show below, the current can be obtained by differentiating the free Dirac Lagrangean $S$, defined by
%\marginpar{Sign?}
\begin{equation*}
 S_{(SM, P, \bar X)}(f)(B) \doteq \int f(x) (D^2_{(SM, P, \bar X)})^{a \alpha}_{\ \ \ b \beta} (B_2)_{a\alpha}^{\ \ \ b\beta}(x,x) \ud_{\bar g}x,
\end{equation*}
\wrt the background connection. Here $ f \in C^\infty_c(M)$ and $D^2$ denotes the Dirac operator acting on the second coordinate.

\begin{remark}
\label{rem:FieldMorphism}
It follows from \eqref{eq:F_Morphism} that, for a multilocal (quantum) field $\Phi$ and a morphism $\chi: (SM, P, \bar X') \to (SM, P, \bar X)$, we have
\[
 \Phi_{(SM, P, \bar X)}(\chi_* t_1, \dots, \chi_* t_k)(B) = \Phi_{(SM, P, \bar X')}(t_1, \dots, t_k)(\chi^* B).
\]
\end{remark}

%\begin{remark}
%\label{rem:star_of_fields}
%If $\Phi$, $\Psi$ are $k$, $l$-local quantum fields, then
%\[
% (\Phi \star \Psi)(t_1, \dots, t_{k+1}) \doteq \Phi(t_1, \dots, t_k) \star \Psi(t_{k+1}, \dots, t_{k+l})
%\]
%is a $k+l$ local quantum field.
%\end{remark}

As we have seen, the construction of fields is straightforward. However, the construction of quantum fields requires a parametrix:
\begin{definition}
\label{def:Parametrix}
A \emph{parametrix} $H$ is a quasi-covariant assignment $(SM, P, \bar X) \to H \in \Gamma_c^\infty(U, D^\oplus_\rho M^2)'$, where $U$ is a neighborhood of the diagonal of $M^2$, such that \eqref{eq:HadamardAnticommutator}, \eqref{eq:HadamardConjugation}, \eqref{eq:HadamardWF} hold.
Quasi-covariance means that for $\chi: D_\rho^\oplus M \to D_\rho^\oplus M'$ the bundle morphism corresponding to a morphism $(SM, P, \bar X) \to (SM', P', \bar X')$ we have that $H - \chi^* H'$ is smooth on the common domain and vanishing at the diagonal, together with all the derivatives.
\end{definition}

A construction prescription for parametrices was given in \cite{LocCovDirac}.
%In \cite{LocCovDirac}, a construction prescription for parametrices was given. For our purposes, it is advantageous to slightly modify this construction. However, as the changes are technical, we refer the interested reader to Appendix~\ref{app:ParametrixConstruction}.

Given a parametrix $H$, we may associate to a local functional $F \in \F_\loc$ an element of $\A_\loc$ by
\begin{equation}
\label{eq:F_to_A}
 (\alpha_H(F))_\omega \doteq \exp(\hbar \Gamma_{\omega - H}) F,
\end{equation}
where the operator $\Gamma$ was defined in \eqref{eq:DefGamma}. This is well-defined as $H - \omega$ is smooth \cite{LocCovDirac} and the values of all its derivatives on the diagonal are unambiguous. As we only act on local functionals, the expression is well-defined even though $H$ is only defined in a neighborhood of the diagonal.

Let us now discuss scaling properties.
%For this, we have to restrict to $M$ with trivial topology, i.e., we consider the local theory. In a given trivialization, the Dirac operator is given by
%\[
% D = - \gamma^\mu (\bar \nabla^S_\mu + \bar A_\mu) + \bar m,
%\]
%where $\bar \nabla^S$ is the spin connection and $\bar A_\mu$ the vector potential.
%We may now consider a different background, with $\bar g'_{\mu \nu} = \lambda^{-2} g_{\mu \nu}$, $\bar A'_\mu = \lambda \bar A_\mu$, $\bar m' = \lambda \bar m$.
One notices that on a scaled background, $\bar X_\lambda \doteq (\lambda^{-2} g, \bar A, \lambda \bar m)$, the Dirac operator also scales. It follows that
%Then
there is a $*$-isomorphism $\sigma_\lambda: \A(SM, P, \bar X_\lambda) \to \A(SM, P, \bar X)$, acting on linear fields as
\begin{equation}
\label{eq:Def_sigma}
 \sigma_\lambda (\psi(u))_\omega = \lambda^{-\frac{n+1}{2}} \psi(f)_{\omega_\lambda},
\end{equation}
where $\omega_\lambda(u,v) = \lambda^{-n-1} \omega(u,v)$, \cf \cite[Lemma~4.2]{HollandsWaldWick} for a proof in the scalar case. For a multilocal (on-shell) quantum field $\Phi$, one defines another multilocal (on-shell) quantum field $S_\lambda \Phi$ by
\begin{equation}
\label{eq:ScaledField}
 (S_\lambda \Phi)_{(SM, P, \bar X)}(t_1, \dots, t_k) \doteq \lambda^{nk} \sigma_\lambda(\Phi_{(SM, P, \bar X_\lambda)}(t_1, \dots, t_k)).
\end{equation}
The \emph{scaling dimension} of a field $\Phi$ is defined as $\frac{n-1}{2}$ times the grade plus the scaling dimension of the geometric factors in $\Phi$ under the scaling $\bar X \to \bar X_\lambda$.
%$(\bar g, \bar A, \bar m) \to (\lambda^{-2} \bar g, \bar A, \lambda \bar m)$.

\subsection{Background variation}
\label{sec:BackgroundVariation}

In the following, we will need 
%To make this more precise, one needs
a means to compare functionals defined on different backgrounds. 
%The main aim of this work is the study of the effect of pertubations of the background.
Consider an object $(SM, P, \bar X)$ of $\CatGSpMan$. We may obtain another object $(SM, P, \bar X + X)$ by considering compactly supported perturbations $X = (g, A, m) \in \mathfrak{P}(SM, P)$, where
\[
  \mathfrak{P}(SM, P) \doteq \Gamma^\infty_c(M, \Sym^2 TM \oplus E^1 \oplus (M \times \R)).
\]
Here $g$ denotes the perturbation of the metric (which has to be chosen small enough in order to preserve the signature), $A$ denotes the variation of the gauge connection, and $m$ denotes the variation of the Yukawa potential. Note that the space of connections is an affine space, so indeed the deviations from a given connection can be parametrized by a vector space, the sections of $E_1$.

We will need to identify configurations or test tensors on $(SM, P, \bar X)$ and $(SM, P, \bar X + X)$. For test tensors, i.e., elements of $\T_c(SM, P)$, this is no problem for tangent vectors and one-forms, as they do not depend on the background structure. Only spinors may be problematic, as the spin structure depends on the metric. To deal with this, we proceed as follows: When changing $(SM, P, \bar X)$ to $(SM, P, \bar X + X)$, we keep the spin and Dirac bundle, and just change the projection from the spin bundle $SM$ to the orthonormal frame bundle $FM$. For that, identify $FM$ with a principal $SO_0(n-1, 1)$ bundle $LM$. To construct $(SM, P, \bar X + X)$, keep $LM$ and the projection from $SM$ to $LM$, but change the identification of $FM$ and $LM$. Fix a trivialization, so that $LM|_U \simeq U \times SO_0(n-1, 1)$, i.e., we are dealing with matrices $B$ such that $B \eta B^t = \eta$, where $\eta = \Diag(-1,1,\dots,1)$. The identification of $LM$ and $FM$ for a metric $h$ is then a map $\pi_h$ such that $\pi_h(B) h \pi_h(B)^t = \eta$, i.e., it is given by a vielbein, $\pi_h(B)_a^\mu = B_a^b e_b^\mu$, with $h^{\mu \nu} = e_a^\mu \eta^{ab} e_b^\nu$. Changing the background metric then amounts to changing the vielbein. Infinitesimally, this change is given by $\delta e^\mu_a = - \frac{1}{2} e^\nu_a h^{\mu \lambda} \delta h_{\nu \lambda}$. This provides an identification of sections of the Dirac bundle on $(SM, P, \bar X)$ and $(SM, P, \bar X')$. This also gives an identification of the corresponding test tensor spaces $\T_c$ and configuration spaces $\wedge \E^\oplus$. This procedure is the one used in  \cite{ForgerRomer} to compute the stress-energy tensor for Dirac fermions.

%\begin{remark}
%\label{rem:Identification}
%Consider $(SM, P, \bar X)$ and $(SM, P, \bar X')$, where the backgrounds differ in a compactly supported region. A test tensor $t \in \T_c(SM, P, \bar X)$ can also be interpreted as a test tensor in $\T_c(SM, P, \bar X')$. As tangent vectors and one-forms do not depend on the background structure, these pose no problem. A bit more involved are spinors, as the definition of the spinor bundle depends on the metric. We use the procedure described in \cite{ForgerRomer}: Locally, fix a basis of the Dirac bundle
%
%Keep Spin and Dirac bundle, just change the projection from the spin to the orthonormal frame bundle $FM$. For that, consider $FM$ as a principal $SO(1, n-1)$ bundle $LM$. The projection $SM \to LM$ is not changed. What is changed is the identification of $LM$ and $FM$. Fix a trivialization, so that $LM = U \times SO(1, n-1)$, i.e., we are dealing with matrices $A$ such that $A \eta A^t = \eta$. The identification of $LM$ and $FM$ is then a map $\pi_g$ such that $\pi_g(A) g \pi_g(A)^t = \eta$, i.e., it is given by a vielbein: $\pi_g(A)_a^\mu = A_a^b e_b^\mu$, with $g^{\mu \nu} = e_a^\mu \eta^{ab} e_b^\nu$. Changing the background metric then amounts to changing the vielbein. Infinitesimally, $\delta e^\mu_a = - \frac{1}{2} e^\nu_a g^{\mu \lambda} \delta g_{\nu \lambda}$. This provides an identification of sections of the Dirac bundle on $(SM, P, \bar X)$ and $(SM, P, \bar X')$. This gives also an identification of the corresponding test tensor spaces $\T_c$ and configuration spaces $\wedge \E^\oplus$.
%\end{remark}

In the following, we will need to consider families of backgrounds $\bar X_s$, $s \in I$, where $I$ is an interval of $\R$. We will require that the modifications only affect a compact subset, i.e., $\bar X_s = \bar X_{s'}$ outside a compact set $K$ and for all $s, s'$. The family is assumed to be smooth in the sense that for any fixed $s_0$, $\bar X_s - \bar X_{s_0}$ is jointly smooth in $s$ and the spacetime point $x$. A \emph{compatible family of functionals} is a family $F_s$, $s \in I$ such that $F_s \in \F(SM, P, \bar X_s)$. The definition of compatible families of on-shell and/or quantum functionals is completely analogous.
\begin{definition}
\label{def:F_s}
Given $F \in \F(SM, P, \bar X_{s_0})$,
\[
 \check F_s(B) \doteq i_{\bar X_s, \bar X_{s_0}} F(B) \doteq F(i_{\bar X_{s_0},\bar X_s} B)
\]
defines a compatible family of functionals, where $i_{\bar X_{s_0},\bar X_s}$ is the bijection from $\wedge \E^\otimes(SM, P, \bar X_{s})$ and $\wedge \E^\otimes(SM, P, \bar X_{s_0})$ described above.
\end{definition}
%The easiest way to construct such a family is to take some $F \in \F(SM, P, \bar X_{s_0})$ and define $F_s(B) \doteq F(B)$, where we used that by the above we have a canonical bijection between $\wedge \E^\otimes(SM, P, \bar X_{s_0})$ and $\wedge \E^\otimes(SM, P, \bar X_s)$.
Note that this construction is not possible for on-shell or quantum functionals.

In order to have an identification of functionals on different backgrounds that is applicable also for on-shell or quantum functionals, we use M{\o}ller operators \cite{BreDue}. Consider two backgrounds $(SM, P, \bar X)$ and $(SM, P, \bar X')$, where the background fields $\bar X$ and $\bar X'$ differ in some compact region $K$. On $\E^\oplus(SM, P)$, we define the retarded M{\o}ller operator as
\[
 r^{\bar X, \bar X'} u \doteq i_{\bar X, \bar X'} u + S^{\oplus}_\ret([ i_{\bar X, \bar X'} \circ {D'}^\oplus - D^\oplus \circ i_{\bar X, \bar X'}] u).
\]
This is well defined, as the expression in square brackets is compactly supported. We have
\begin{align}
\label{eq:MollerMap}
 D^\oplus (r^{\bar X, \bar X'} u) & = i_{\bar X, \bar X'} {D'}^\oplus u, & \supp ( r^{\bar X, \bar X'} u - i_{\bar X, \bar X'} u ) \subset J^+(K).
\end{align}
Obviously, this map is continuous, and commutes with the conjugation $^*$. It is invertible, the inverse being given by $r^{\bar X', \bar X}$. We denote its transpose \wrt the pairing $\skal{ \cdot}{\cdot} : \D^\oplus(SM,P) \times \E^\oplus(SM,P) \to \C$ by ${r^{\bar X, \bar X'}}^t$.

On $\wedge^k \E^\oplus(SM, P)$, the M{\o}ller operator is defined as the continuous linear operator which acts on elements $u_1 \wedge \dots \wedge u_k$, $u_i \in \E^\oplus(SM, P)$ as
\[
 r^{\bar X, \bar X'} (u_1 \wedge \dots \wedge u_k) = (r^{\bar X, \bar X'} u_1) \wedge \dots \wedge (r^{\bar X, \bar X'} u_k).
\]

The retarded M{\o}ller operator can be used to construct a $*$-isomorphism of the algebras $\A(SM, P, \bar X')$ and $\A(SM, P, \bar X)$:
\begin{definition}
\label{def:tau}
For $F \in \A(SM, P, \bar X')$ we set
\[
 (\tau_\ret^{\bar X, \bar X'} F)_\omega(B) \doteq F_{\omega'}(r^{\bar X', \bar X} B),
\]
where $\omega'( \cdot, \cdot) \doteq \omega({r^{\bar X', \bar X}}^t \cdot, {r^{\bar X', \bar X}}^t \cdot)$.
\end{definition}

On elements of $\F(SM, P, \bar X')$, $\tau_\ret$ is defined in complete analogy.

\begin{remark}
To see that $\omega'$ as defined in Definition~\ref{def:tau} is a Hadamard two-point function, we first note that ${r^{\bar X', \bar X}}^t \circ {D'}^\oplus = D^\oplus \circ i_{\bar X, \bar X'}$, so that $\omega'$ is a bi-solution. As the M{\o}ller operator commutes with the conjugation, also \eqref{eq:HadamardConjugation} holds. Furthermore, $\omega'$ and $\omega$ coincide outside the causal future of $K$. In particular, they coincide when restricted to a neighborhood of a Cauchy surface in the past of $K$, so that $\omega'$ fulfills \eqref{eq:HadamardAnticommutator} and \eqref{eq:HadamardWF} there. As $\omega'$ is a bisolution and has the correct Cauchy data, it fulfills \eqref{eq:HadamardAnticommutator} everywhere. That \eqref{eq:HadamardWF} holds everywhere follows from arguments given in \cite[Sec.~4.2]{SandersDirac}.
\end{remark}

\begin{proposition}
\label{prop:tau_isomorphism}
%The map $\tau_\ret^{\bar X, \bar X'}: \A(SM, P, \bar X') \to \A(SM, P, \bar X)$ is a $*$-isomorphism which restricts to a $*$-isomorphism $\A_S(SM, P, \bar X') \to \A_S(SM, P, \bar X)$.
The map $\tau_\ret^{\bar X, \bar X'}$ is a $*$-isomorphism, which restricts to a $*$-isomorphism of the algebras of on-shell (quantum) functionals.
\end{proposition}
\begin{proof}
The homomorphism property follows from the form \eqref{eq:star_omega} of the $\star_\omega$ product and the definition of $\omega'$. The isomorphism property follows from the invertibility of the M{\o}ller operator, and the $*$ property by the fact that the involution $^*$ commutes with the M{\o}ller operator. Continuity follows from the M{\o}ller operator being continuous. The last statement follows from the fact that ${D'}^\oplus (r^{\bar X', \bar X} u) = 0$ if and only if $D^\oplus u = 0$.
\end{proof}

The restriction of $\tau_\ret$ to on-shell quantum functionals coincides with the retarded variation employed in \cite{HollandsWaldStress}.

\begin{proposition}
\label{prop:tauMorphism}
If $(SM, P, \bar X)$ and $(SM, P, \bar X')$ as above are related by a morphism $\chi: (SM, P, \bar X') \to (SM, P, \bar X)$ of $\CatGSpMan$ and $\Phi$ is a multilocal on-shell (quantum) field, then
\[
 \tau_\ret^{\bar X, \bar X'} (\Phi_{(SM, P, \bar X')}(t_1, \dots, t_k)) = \Phi_{(SM, P, \bar X)}(\chi_* t_1, \dots, \chi_* t_k).
\] 
\end{proposition}
\begin{proof}
As a straightforward consequence of \eqref{eq:MollerMap} and the uniqueness of the Cauchy problem, $r^{\bar X', \bar X} u = \chi^* u$ for $D^\oplus u = 0$. Hence, also $\omega' = \chi_* \omega$.
%We have $\chi^* \circ {D'}^{\oplus} = D^\oplus \circ \chi^*$ and $\chi^* \circ S'^{\oplus}_\ret = S_\ret^\oplus \circ \chi^*$.
%%intertwine the action of $D^\oplus$ and $S^\oplus_\ret$.
%Hence, $r^{\bar X', \bar X} B = \chi^* B$ for $B \in \wedge \E^\oplus(SM,P)$ and $\omega' = \chi_* \omega$.
The claim then follows from Remark~\ref{rem:FieldMorphism}.
%The claim thus follows from the fact that fields are natural transformations,
%\[
% \A(\chi) (\Phi_{(SM, P, \bar X')}(t_1, \dots t_k)) = \Phi_{(SM, P, \bar X)}(\chi_* t_1, \dots, \chi_* t_k),
%\]
%and that $\A(\chi)$ acts via pullback of the configuration, $(\A(\chi) F)(B) \doteq F(\chi^* B)$.
%identity
%\[
% \Phi_{(SM, P, \bar X)}(\chi^* t_1, \dots, \chi^* t_n)(\chi^* B) = \Phi_{(SM, P, \bar X')}(t_1, \dots, t_n)(B).
%\]
\end{proof}

It is convenient to have an infinitesimal version of $\tau_\ret$:

\begin{definition}
Let $X \in \mathfrak{P}(SM, P)$. Choose a smooth family $\bar X_s$ such that $\bar X_0 = \bar X$ and $\del_s \bar X_s|_{s=0} = X$. Given a compatible family $F_s$ of (on-shell) (quantum) functionals, define
\[
 \delta^X_\ret F \doteq \frac{\ud}{\ud s} \tau^{\bar X, \bar X_s} F_s|_{s=0}.
\]
\end{definition}

As a straightforward consequence of Proposition~\ref{prop:tau_isomorphism}, it follows that $\delta_\ret$ fulfills, for compatible families $F, G$ of quantum functionals,
\begin{align}
\label{eq:delta_Leibniz}
 \delta_\ret^X (F \star G) & = \delta_\ret^X F \star G + F \star \delta_\ret^X G, \\
\label{eq:delta_hermitean}
 \delta_\ret^X F^* & = (\delta_\ret^X F)^*,
\end{align}
where
\begin{align*}
 (F \star G)_s & \doteq F_s \star_s G_s, &
 (F^*)_s \doteq F_s^*.
\end{align*}

For fields, yet another way to construct a compatible family of functionals is possible: Given a $k$-local (on-shell) (quantum) field $\Phi$ and a set of corresponding test tensors $t_k \in \T_c(SM, P, \bar X_{s_0})$, we may define
\begin{equation}
\label{eq:TildePhi}
 \tilde \Phi(t_1, \dots, t_k)_s \doteq \Phi_{(SM, P, \bar X_s)}(i_{\bar X_s, \bar X_{s_0}} t_1, \dots, i_{\bar X_s, \bar X_{s_0}} t_k).
\end{equation}
%using the canonical bijection between $\T_c(SM, P, \bar X_{s_0})$ and $\T_c(SM, P, \bar X_s)$.
%A similar construction for quantum fields is possible, using the M{\o}ller operators introduced below.
A straightforward consequence of Proposition~\ref{prop:tauMorphism} is then:
%\marginpar{Check}
\begin{lemma}
\label{lemma:deltaMorphism}
Let $\Phi$ be a multilocal on-shell (quantum) field and $\chi_s$ be a family of morphisms $(SM, P, \bar X_s) \to (SM, P, \bar X)$ with $\bar X = \bar X_0$ and $\bar X_s - \bar X$ compactly supported. Let $X = \del_s \bar X_s|_{s=0}$. Then
\[
 \delta_\ret^X \tilde \Phi(t_1, \dots, t_k) = \sum_j \Phi_{(SM, P, \bar X)}(t_1, \dots, \Lie_X t_j, \dots t_k),
\]
%\[
% \delta_\ret \Phi_{(SM, P, \bar X)}(X, t_1, \dots, t_k) = - \sum_j \Phi_{(SM, P, \bar X)}(t_1, \dots, \Lie_X t_j, \dots t_k),
%\]
where $\Lie_X t \doteq \frac{\ud}{\ud s} \chi_{s *} ( i_{\bar X_s, \bar X} t )|_{s=0}$.
\end{lemma}

In the following, we require a smooth dependence of fields on the background. Denoting by $\tilde \Phi(x)_s$ the distributional kernel of $\Phi(\cdot)_s$, consider the kernel of $(i_{\bar X_{s_0}, \bar X_s} \tilde \Phi(x)_s)_k$ as a distributional section on $I \times M^{k+1}$. We require that its wave front set is contained in
\begin{multline*}
 \left\{ (s, x, x_1, \dots, x_k; \sigma, \xi, \xi_1, \dots, \xi_k) \in \dot T^*(I \times M^{k+1}) \mid \right. \\
 \left. x=x_1 = \dots = x_k, \xi + \sum\nolimits_i \xi_i = 0, (\xi, \xi_1, \dots, \xi_k) \neq 0 \right\}.
\end{multline*}
In the case of an analytic family of analytic backgrounds, the same should be true for the analytic wave front set, \cf \cite{HollandsWaldTO, HoermanderI}.

Knowing that our fields are smooth \wrt variations of the background, we may consider the derivative of fields \wrt to such variations:

\begin{definition}
Let $X \in \mathfrak{P}(SM, P)$. Choose a smooth family $\bar X_s$ such that $\bar X_0 = \bar X$ and $\del_s \bar X_s|_{s=0} = X$. Then the functional differentiation of a field $\Phi$ \wrt the background is defined as
\begin{equation}
\label{eq:DefFieldDerivative}
 \Phi^{(1)}_{(SM, P, \bar X)}(X,t) \doteq \del_s (i_{\bar X, \bar X_s} \tilde \Phi(t)_s )|_{s=0}.
\end{equation}
%\[
% \Phi^{(1)}_{(SM, P, \bar X)}(X,t)(B) \doteq \del_s \tilde \Phi(t)_s(B)|_{s=0}.
%\]
%\[
%  \Phi^{(1)}_{(SM, P, \bar X)}(X,t_1, \dots, t_k)(B) \doteq \del_s \tilde \Phi(t_1, \dots, t_k)_s(B)|_{s=0}.
%\]
\end{definition}
%In order to differentiate \wrt variations of the background $\bar X$ in the direction $X \in \Gamma^\infty_c(M, \Sym^2 TM \oplus E^1 \oplus \R)$, one chooses a smooth family $\bar X_s$ such that $\bar X_0 = \bar X$ and $\del_s \bar X_s|_{s=0} = X$. We may then define the functional differentiation of a field \wrt the background as
%\[
% \Phi^{(1)}_{(SM, P, \bar X)}(X,t)(B) \doteq \del_s \tilde \Phi(t)_s(B)|_{s=0}.
%\]

%For a two-parameter family $\bar X_{s,t}$ of variations, the definition of $\tilde \Phi$ is completely analogous. We may then define
%\[
% \Phi^{(2)}_{(SM, P, \bar X)}(X,t)(B) \doteq \del_s \tilde \Phi(t)_s(B)|_{s=0}.
%\]

%Having an identification of test tensors from $\T_c(SM, P, \bar X)$ and $\T_c(SM, P, \bar X + X)$, we may define 
%
%
%$\wedge \E^\otimes(SM, P, \bar X)$ and $\wedge \E^\otimes(SM, P, \bar X + X)$ and of $\T_c(SM, P, \bar X)$ and $\T_c(SM, P, \bar X + X)$, we may now compare the same field on different backgrounds:
%
%\begin{definition}
%\label{def:FunctionalDiffField}
%The functional differentiation of a field \wrt the background field $\bar X$ is defined as
%\[
% \Phi^{(1)}_{(SM, P, \bar X)}(X,t)(B) \doteq \frac{\ud}{\ud s} \Phi_{(SM, P, \bar X + sX)}(t)(B) |_{s=0}.
%\]
%For a quantum field, one defines
%\[
% \Phi^{(1)}_{(SM, P, \bar X)}(X,t)_\omega(B) \doteq \frac{\ud}{\ud s} \Phi_{(SM, P, \bar X + sX)}(t)_\omega(B) |_{s=0}.
%\]
%\end{definition}

By definition we have that $\del_s (i_{\bar X, \bar X_s} \check \Phi_{(SM, P, \bar X)}(t)_s ) = 0$, \cf Definition~\ref{def:F_s}. An obvious consequence is that
\begin{equation}
\label{eq:Diff_tilde_Phi_Phi}
 \del_s \left( i_{\bar X, \bar X_s} \left( \tilde \Phi(t)_s - \check \Phi_{(SM, P, \bar X)}(t)_s \right) \right)|_{s=0} = \Phi^{(1)}_{(SM, P, \bar X)}(X, t).
\end{equation}

We will impose one more condition on the smoothness of fields: Leaving the two slots for test tensors open, we may consider the kernel of $\Phi^{(1)}_{(SM,P,\bar X)}(\cdot, \cdot)_k$ as a distributional section on $M^{k+2}$. From \eqref{eq:Field}, it is clear that it has support on the total diagonal $\diag_{k+2}$. We also require it to have finite order and wave front set orthogonal to $T \diag_{k+2}$.
%contained in
%\begin{equation*}
% \left\{ (x, y, x_1, \dots, x_k; \xi, \eta, \xi_1, \dots, \xi_k) \in \dot T^* M^{k+2} \mid \xi + \eta + \sum\nolimits_i \xi_i = 0 \right\}.
%\end{equation*}
In particular, $\Phi^{(1)}$ can be seen as a field, where the test tensor is some linear combination of $\bar \nabla_{(\lambda_1} \dots \bar \nabla_{\lambda_r)} X \otimes \bar \nabla_{(\rho_1} \dots \bar \nabla_{\rho_s)} t$. We denote this new test tensor by $X * t$ in the following, so that, when we wish to emphasize that $\Phi^{(1)}$ is a field, we write $\Phi^{(1)}_{(SM, P, \bar X)}(X * t)$ instead of $\Phi^{(1)}_{(SM, P, \bar X)}(X,t)$.

%\begin{remark}
%\label{rem:Phi_1_Field}
%It follows from the fact that fields are natural transformations, that $\Phi_{(SM, P, \bar X)}(t)$ depends on the background $\bar X$ only on the support of $t$. By choosing the support of $t$ arbitrarily small, we conclude that
%%the part of $\Phi$ that depends on the jet of $t$ at $x$ only depends on the jet of the perturbation $X$ at $x$.
%$\Phi^{(1)}_{(SM,P,\bar X)}(\cdot, \cdot)$ as an $\F(SM,P)$ valued distributional section on $M^2$ has support on the diagonal.
%Hence, $\Phi^{(1)}$ is again a field, where the test tensor is some linear combination of $\bar \nabla_{(\lambda_1} \dots \bar \nabla_{\lambda_r)} X \otimes \bar \nabla_{(\rho_1} \dots \bar \nabla_{\rho_s)} t$. We denote this new test tensor by $X * t$ in the following, so that, when we wish to emphasize that $\Phi^{(1)}$ is a field, we write $\Phi^{(1)}_{(SM, P, \bar X)}(X * t)$ instead of $\Phi^{(1)}_{(SM, P, \bar X)}(X,t)$.
%%The same applies for quantum fields.
%\end{remark}

\begin{definition}
\label{def:Current}
The \emph{current} is given by
\begin{equation}
\label{eq:Def_j}
 j_{(SM, P, \bar X)}(A) \doteq S^{(1)}_{(SM, P, \bar X)}(A, f),
\end{equation}
with $f$ chosen to be identical to $1$ on the support of $A$. By the above, this defines a field. % (quantum) field.
This coincides with \eqref{eq:DefCurrent}.
\end{definition}

Finally, we want to consider how $\delta^X_\ret$ behaves under scaling, \cf \eqref{eq:Def_sigma}.
From the definition of $r^{\bar X, \bar X'}$ and the fact that the Dirac operator scales under the scaling transformation $\bar X \to \bar X_\lambda = (\lambda^{-2} \bar g, \bar A, \lambda \bar m)$, we conclude that $r^{\bar X, \bar X'} = r^{\bar X_\lambda, \bar X'_\lambda}$, and hence
\begin{equation}
\label{eq:delta_sigma}
 \delta^X_\ret \sigma_\lambda = \sigma_\lambda \delta^{X_\lambda}_\ret
\end{equation}
%\marginpar{Check}
on a family of (on-shell) (quantum) functionals compatible with a family of backgrounds $\bar X_s$ with $\bar X_0 = \bar X_\lambda$ and $\del_s \bar X_s |_{s=0} = X_\lambda = (\lambda^{-2} g, A, \lambda m)$.

\subsection{Time-ordered products}

In order to describe interacting fields, one has to introduce time-ordered products. These are natural transformations from $\F_\loc^{\otimes k}$ to $\A^\os$, interpreted as functors from $\CatGSpMan$ to $\CatVec_i$, for arbitrary $k$, fulfilling certain axioms.
%Note that we do not require the involution to be preserved, instead we have the axiom of {\bf Unitarity}.
By concatenation with fields, the time-ordered products are thus multilocal on-shell quantum fields.\footnote{For some purposes, it is conventient to consider time-ordered products that map into the off-shell quantum fields, \cf \cite{BDF09, RejznerFredenhagenQuantization}, for example. For our purposes, the on-shell quantum fields are sufficient.}
%\marginpar{Apart from involution}
The axioms for time-ordered products are \cite{HollandsWaldTO}:
\begin{description}

\item[Starting element:]
On a c-number functional, i.e., a functional satisfying $F(B) = F_0$ for all $B \in \wedge \E^{\oplus}$, $\TO$ acts as the identity.
 
\item[Symmetry:]
 Time-ordered products are graded symmetric, i.e.,
 \[
 \TO(F_1, \dots, F_i, F_{i+1}, \dots, F_k) = (-1)^{\betrag{F_i} \betrag{F_{i+1}}} \TO(F_1, \dots, F_{i+1}, F_i, \dots, F_k).
 \]

\item[Support:]
 $\supp \TO (F_1, \dots, F_k) \subset \cup_i \supp F_i$.
 
\item[Causal factorization:] Let $F_i$, $G_j$ be such that $\supp F_i \cap J^-(\supp G_j) = \emptyset$ for all $i, j$. Then
 \[
  \TO(F_1, \dots, F_k, G_1, \dots G_l) = \TO(F_1, \dots F_k) \star \TO(G_1, \dots G_l).
 \]

\item[Scaling:] The time-ordered products scale \emph{almost homogeneously}, i.e., if applied to fields $\Phi_i$ we have
\begin{multline}
\label{eq:Scaling}
 \lambda^{-d_\Phi} S_\lambda \TO(\Phi_1(t_1), \dots \Phi_k(t_k)) = \\
\sum_{I_0 \sqcup \dots \sqcup I_j } (-1)^\Pi \TO( \Phi_{I_0}(t_{I_0}), r_\lambda( \Phi_{I_1}(t_{I_1}) ), \dots, r_\lambda( \Phi_{I_j}(t_{I_j}) )).
\end{multline}
Here the sum is over all partitions of $\{ 1, \dots, k \}$ into disjoint subsets, with $I_i \neq \emptyset$ for all $i \geq 1$. $\Phi_I$ stands for the collection of $\Phi_i$ with $i \in I$,
%and $*t_I$ for the restriction to the diagonal of covariant derivatives of the $t_i$, for $i \in I$, analogously to Remark~\ref{rem:Phi_1_Field}.
and $\Pi$ is a combinatorial factor which takes the grades and permutations of the $\Phi_i$ into account. $d_\Phi$ is the sum of the scaling dimensions of the $\Phi_i$. Finally, $r_\lambda$ are natural transformations from $\F^{\otimes m}_\loc$ to $\F_\loc$, which fulfill the properties of renormalization maps discussed in Remark~\ref{rem:Renormalization} below, and which are polynomials in $\log \lambda$.

\item[Microlocal spectrum condition:]  Let $\omega$ be a Hadamard state. % on $\A_S(SM, P, \bar X)$.
Then for all fields $\Phi_i$ the distributional section $\omega(\TO( \Phi_1(x_1), \dots, \Phi_k(x_k))$ has a wave front set contained in $C_T^k \subset T^* M^k$, defined through decorated graphs, \cf \cite{BrunettiFredenhagenScalingDegree, HollandsWaldTO}.

\item[Smoothness:] The time-ordered products depend smoothly on the background fields \cite{HollandsWaldTO}. Thus, let $\bar X_s$ depend
%\footnote{When deforming the background metric, one of course has to specify how the spin structure is treated. But spin structures are classified by $H^1(M, \Z^2)$, i.e., to each homotopy class of non-contractible paths, one associates a sign, indicating whether one changes the sheet in the two-fold cover of the Lorentz group. When deforming the metric, one does not change this assignment.}
smoothly on a parameter $s \in \R$. Let $\omega^{(s)}$ be a family of Hadamard states on $\A(SM, P, \bar X_s)$,
%with smooth truncated $n$-point functions,
smoothly depending on $s$.
One then requires that, for all fields $\Phi_i$,
%\begin{multline*}
% \WF \left( \omega^{(s)} \left( \TO^{(s)}(F_1(x_1), \dots, F_k(x_k)) \right) \right) \\ \subset \left\{ (s, \sigma; \{ x_i, \xi_i \}) \in \dot T^*(\R \times M^k) | (\{x_i, \xi_i \}) \in C_T^{k, (s)} \right\}.
%\end{multline*}
\begin{multline*}
 \WF \left( \omega^{(s)} \left( \TO^{(s)}(\tilde \Phi_1(x_1)_s, \dots, \tilde \Phi_k(x_k)_s) \right) \right) \\ \subset \left\{ (s, \sigma; \{ x_i, \xi_i \}) \in \dot T^*(\R \times M^k) | (\{x_i, \xi_i \}) \in C_T^{k, (s)} \right\},
\end{multline*}
where we used \eqref{eq:TildePhi}.

\item[Analyticity:] In the case of an analytic spacetime, the time-ordered products depend analytically on the background fields. This is made precise by a condition analogous to the one for smoothness, \cf \cite{HollandsWaldTO}.

\item[Expansion:] The time ordered product commutes with functional differentiation, i.e.,
\begin{equation}
\label{eq:TO_Expansion}
 \TO(F_1, \dots, F_k)^{(1)}(x) = \sum_{i = 1}^k (-1)^{\sum_{l=1}^{i-1} \betrag{F_l}} \TO(F_1, \dots, F_i^{(1)}(x), \dots, F_k).
\end{equation}

\item[Unitarity:] We have
\begin{equation}
\label{eq:Unitarity}
\TO (F_1, \dots, F_k)^* = \sum_{I_1 \sqcup \dots \sqcup I_j} (-1)^{k+j+\Pi} \TO( F_{I_1}^*) \star \dots \star \TO( F_{I_j}^*),
\end{equation}
where $I_1 \sqcup \dots \sqcup I_j$ denotes all partitions of $\{ 1, \dots, k \}$ into nonempty, pairwise disjoint subsets. $\TO(F^*_I)$ stands for $\TO(F^*_{I(i)}, \dots, F^*_{I(1)})$ and $\Pi$ denotes a combinatorial factor, depending on the grades of the $F_i$ and the partition, which accounts for the reordering of the $F_i$ on the right hand side.

\item[Source term:] For a linear functional $F(B) = \int \skal{f}{B_1}(x) \ud_{\bar g}x$, with $f \in \D^\oplus_\rho(SM,P)$, we have
\begin{multline}
\label{eq:T11a}
 \TO(F, F_1, \dots, F_k) = F \star \TO(F_1, \dots, F_k) \\
+ i \hbar \sum_j (-1)^{\sum_{l=1}^{j-1} \betrag{F_l}} \TO(F_1, \dots, F^{(1)}_j(S^\oplus_\ret f), \dots, F_k).
\end{multline}

%\item[Equation of motion:] If $\psi$ denotes the linear field \eqref{eq:LinearField}, then
%\begin{equation}
%\label{eq:TO_eom}
% T( \psi(D^\oplus u) \gt F_1 \gt \dots \gt F_k) = i \skal{T(\hat \bigotimes_i F_i)^{(1)}}{u} + \psi(D^\oplus u) \star T(\hat \bigotimes_i F_i).
%\end{equation}
\end{description}

Given a parametrix $H$, the map $\alpha_H$ defined in \eqref{eq:F_to_A} defines time-ordered products for $k=1$.

We note that condition \eqref{eq:T11a} corresponds to axiom T11a of \cite{HollandsWaldStress}. It implies the axiom of the field equation, which was used in \cite{LocCovDirac}. The proof that it can be fulfilled by a redefinition of the time-ordered products proceeds as in the scalar case, \cf \cite[Section~6.1]{HollandsWaldStress}.

\begin{remark}
\label{rem:Renormalization}
The time-ordered products are not unique. Given time-ordered products $\TO$, one can define new ones, $\TO'$, by
\[
 \TO'(F_1, \dots, F_k) = \sum_{I_0 \sqcup \dots \sqcup I_j } (-1)^\Pi \TO( F_{I_0}, r^{\betrag{I_1}}( F_{I_1} ), \dots, r^{\betrag{I_j}}( F_{I_j} )),
\]
where we used the same notation as in \eqref{eq:Scaling} and $r^k$, $k \leq 1$ are natural transformations from $\F_\loc^{\otimes k}$ to $\F_\loc$. They are at least  of first order in $\hbar$, lower the degree in the total grade by even numbers, and fulfill {\bf Symmetry}, {\bf Support}, {\bf Scaling}, {\bf Expansion}. As we can see a local functional as a field, we may write, for arbitrary fields $\Phi_i$,
\[
 r^k(\Phi_1(x_1), \dots, \Phi_k(x_k)) = \int C_k(x,x_1, \dots x_k) \Phi(x) \ud_{\bar g} x,
\]
where $\Phi$ is a field and $C$ a distributional section on $M^{k+1}$, supported on $\diag_{k+1}$. Of course both $C$ and $\Phi$ depend on the $\Phi_i$. In order to ensure the {\bf Microlocal spectrum condition}, we require
\[
 \WF(C_k) \perp T\diag_{k+1}.
\]
%consider the integral kernel of the image of $r_k$. By {\bf Support}, $r_k$ is supported on $\diag_{k+1}$. In order to ensure the {\bf Microlocal spectrum condition}, we require
%\[
% \WF(r_k(\Phi_1(x_1), \dots, \Phi(x_k))(x)) \perp T\diag_{k+1}.
%\]
%Note that we see $r_k$ as an $\F_\loc$ valued distributional section of $M^{k+1}$. As $\F_\loc$ is a locally convex vector space, one can define the wave front set, \cf \cite{RejznerFermions}. 
In order to ensure {\bf Smoothness}, one requires that for smooth variations $\bar X_s$ of the background field,
\[
 \WF(C_k^s(x, x_1, \dots, x_k)) \perp T(I \times \diag_{k+1}),
\]
%\[
% WF(r_k^s(\tilde \Phi_1(x_1), \dots, \tilde \Phi(x_k))(x)) \perp T(I \times \diag_{k+1}),
%\]
where $I$ is the interval on which $s$ is varied. Analogously, one defines the {\bf Analyticity} condition. To ensure {\bf Starting element}, $r^1$ has to vanish on c-number functionals, and to ensure {\bf Unitarity} one requires
%\marginpar{Follows from nat trafo if IVec}
\[
 r^k(F_1, \dots, F_k)^* = (-1)^{k+1} r^k(F_k^*, \dots F_1^*).
\]
%\[
% r_k(F_1, \dots, F_k)^* = \sum_{I_1 \sqcup \dots \sqcup I_j} (-1)^{k+j+\Pi} r_{\betrag{I_1}}( F_{I_1}^*) \star \dots \star r_{\betrag{I_j}}( F_{I_j}^*),
%\]
%Do we need a state? Also smoothness and analyticity. What about unitarity?
%\begin{align*}
% r_1 & = \id, \\
% r_k( F_1, \dots, F_{k-1}, \psi(f)) & = \psi(f),
%\end{align*}
Finally, to preserve {\bf Source term}, one requires $r^k$ to vanish if one of its entries is a linear functional.
%One can then define new time-ordered products by
%\[
% \TO'(F_1, \dots, F_k) = \sum_{I_0 \sqcup \dots \sqcup I_j } (-1)^\Pi \TO( F_{I_0}, r_{\betrag{I_1}}( F_{I_1} ), \dots, r_{\betrag{I_j}}( F_{I_j} )),
%\]
%where we used the same notation as in \eqref{eq:Scaling}. %, except for the fact that $I_0$ is allowed to be empty.
This exhausts the renormalization freedom of time-ordered products \cite{HollandsWaldWick}.
\end{remark}

Given an interaction Lagrangean $S^\ia$, i.e., a field $S^\ia_{(SM, P, \bar X)}(f)$ for $f \in C_c^\infty(M)$,  one can define interacting (multilocal) quantum fields as follows: One first restricts attention to an open causally convex region $\Oo \subset M$ and to test sections $t$ supported inside $\Oo$. One then chooses a test function $h$, which is identical to $1$ on the closure of $\Oo$, and sets, for arbitrary fields $\Phi$,
\begin{multline*}
 (\Phi_{S^\ia(h)})_{(SM, P, \bar X)}(t) \\
 \doteq \sum_{k=0}^\infty \frac{i^k}{\hbar^k k!} \Ret(\Phi_{(SM, P, \bar X)}(t); \underbrace{S^\ia_{(SM, P, \bar X)}(h),  \dots,  S^\ia_{(SM, P, \bar X)}(h)}_{k \text{ times}}),
\end{multline*}
where the \emph{retarded product} $\Ret$ is defined by
\begin{align}
\label{eq:Ret}
 \Ret(e^{iF}; e^{iG}) & \doteq \TO(e^{iG})^{\star -1} \star \TO(e^{i(F+G)}), & \TO(e^{iF}) & \doteq \sum_{k=0}^\infty \frac{i^k}{k!} \TO(\underbrace{F, \dots, F}_{k \text{ times}}).
\end{align}
The definition \eqref{eq:Ret} has to be understood in the sense of formal power series in $F$ and $G$. The interacting quantum field is a formal power series in $S^\ia$ (and also $\hbar$). The algebra $\A_{S^\ia}(\Oo)$ generated by the $(\Phi_{S^\ia(h)})_{(SM, P, \bar X)}(t)$'s does not depend on the choice of the cut-off function $h$. By taking the inductive limit $\Oo \to M$, one obtains the full interacting algebra, \cf \cite{BrunettiFredenhagenScalingDegree} for details. This is the \emph{algebraic adiabatic limit}. Time-ordered products taking values in $\A_{S^\ia}(\Oo)$ can be defined by
\[
 \TO_{S^\ia}(F_1, \dots, F_k) \doteq \sum_{j=0}^\infty \frac{i^j}{\hbar^j j!} \Ret(F_1, \dots, F_k; \underbrace{S^\ia(h),  \dots,  S^\ia(h)}_{j \text{ times}}),
\]
where the $F_i$ are supported in $\Oo$.

\section{The principle of perturbative agreement}
\label{sec:PPA}

The principle of perturbative agreement requires that it does not matter whether one puts parts of the free Lagrangean into the interaction Lagrangean. This means that,
%The principle of perturbative agreement now requires that,
for arbitrary local functionals $F_i \in \F_\loc(SM, P, \bar X)$, and arbitrary background fields $\bar X$, $\bar X'$, differing in a compactly supported region, one has
\begin{multline}
\label{eq:PPA}
 \tau_\ret^{\bar X, \bar X'} \TO_{(SM, P, \bar X')} (i_{\bar X', \bar X} F_1, \dots, i_{\bar X', \bar X} F_k) \\ = \Ret_{(SM, P, \bar X)}( F_1, \dots, F_k; e^{i(S_{\bar X'}-S_{\bar X})/\hbar}).
\end{multline}
%Here %$\sim$ denotes equality on-shell 
%the subscripts at $\TO$ and $\Ret$ indicate to which background they correspond. 
Here we used the identification of elements of $\F(SM, P, \bar X')$ and $\F(SM, P, \bar X)$ introduced in Section~\ref{sec:BackgroundVariation}.
%Note that the time-ordered product acts on elements of $\F_\loc(SM, P, \bar X)$, but, as discussed above, these may also be interpreted as elements of $F_i \in \F_\loc(SM, P, \bar X')$.
%Note that it is crucial to have fields here, as on the \lhs we deal with functionals in $\F_\loc(SM, P, \bar X')$ and on the \rhs with functionals in $\F_\loc(SM, P, \bar X)$. 

As such, the above equation is not meaningful, as the \rhs is a formal power series in $S_{\bar X'}-S_{\bar X}$, whereas the \lhs is not. To cure this, we consider the infinitesimal version:
\begin{equation*}
 \delta_\ret^X \TO(\check F_1, \dots, \check F_k) = i \hbar^{-1} \Ret(F_1, \dots, F_k; S^{(1)}(X, f)),
\end{equation*}
%\begin{equation*}
% (\delta_\ret \TO(\Phi_1, \dots, \Phi_k))(X, t_1, \dots t_k) \sim i \hbar^{-1} \Ret(\Phi_1(t_1), \dots, \Phi_k(t_k); S^{(1)}(f, X)),
%\end{equation*}
%\begin{multline*}
% (\delta_\ret \TO(\Phi_1, \dots, \Phi_k))(X, t_1, \dots t_k) \sim i \Ret(\Phi_1(t_1), \dots, \Phi_k(t_k); S^{(1)}(f, X)) \\ + \sum_j \TO(\Phi_1(t_1), \dots, \Phi^{(1)}_j(X, t_j), \dots \Phi_k(t_k)),
%\end{multline*}
where on the \rhs one chooses an arbitrary $f$ which is equal to $1$ on $\supp X$. Using \eqref{eq:Diff_tilde_Phi_Phi}, we may write this for fields in a form similar to the one used in \cite{HollandsWaldStress}:
\begin{multline*}
 \delta_\ret^X \TO(\tilde \Phi_1(t_1), \dots, \tilde \Phi_k(t_k)) = i \hbar^{-1} \Ret(\Phi_1(t_1), \dots, \Phi_k(t_k); S^{(1)}(X, f)) \\ + \sum_j \TO(\Phi_1(t_1), \dots, \Phi_j^{(1)}(X, t_j), \dots, \Phi_k(t_k)).
\end{multline*}
For variations of the gauge background, this means
\begin{multline}
\label{eq:PPA_gauge}
 \delta_\ret^A \TO(\tilde \Phi_1(t_1), \dots, \tilde \Phi_k(t_k)) = i \hbar^{-1} \Ret(\Phi_1(t_1), \dots, \Phi_k(t_k); j(A)) \\ + \sum_j \TO(\Phi_1(t_1), \dots, \Phi_j^{(1)}(A, t_j), \dots, \Phi_k(t_k)).
\end{multline}
%\begin{equation}
%\label{eq:PPA_gauge}
% (\delta_\ret \TO(\Phi_1, \dots, \Phi_k))(A, t_1, \dots t_k) \sim i \hbar^{-1} \Ret(\Phi_1(t_1), \dots, \Phi_k(t_k); j(A)).
%\end{equation}
%\begin{multline}
%\label{eq:PPA_gauge}
% (\delta_\ret \TO(\Phi_1, \dots, \Phi_k))(A, t_1, \dots t_k) \sim i \Ret(\Phi_1(t_1), \dots, \Phi_k(t_k); j(A)) \\ + \sum_j \TO(\Phi_1(t_1), \dots, \Phi^{(1)}_j(A, t_j), \dots \Phi_k(t_k)).
%\end{multline}

%\begin{remark}
%\label{rem:DefDelta}
%\marginpar{Check}
%Note that in \cite{HollandsWaldStress}, there is a supplementary term, taking into account that there also the test tensors are allowed to depend on the background fields. We think that the present approach, where such a dependence is not permitted, is conceptually clearer and simplifies some of the proofs.
%\end{remark}

\begin{remark}
\label{rem:GaugeTrafo}
Often, we will be considering gauge transformations, i.e., $A = \bar \ud c$, $c \in \Gamma^\infty_c(M, E^0)$, and make use of Lemma~\ref{lemma:deltaMorphism}. In that case, the family of morphisms maps the background connection\footnote{For nonabelian gauge groups, the family $\chi_s$ of gauge transformations generated by $A = \bar \ud c$ is not affine linear in $s$. However, for the present purposes only the linear component matters, so we stick to this inaccurate notation.} $\bar A + s \bar \ud c$ to the background connection $\bar A$, i.e., $\chi_s$ performs the gauge transformation $\bar A + s \bar \ud c \to \bar A$. The corresponding push-forward action on test tensors is $\chi_{s *} t = t + s \rho(c) t$, i.e., we define $\Lie_c t \doteq \rho(c) t$. For a $k$-local field $\Phi$ we define
\begin{equation}
\label{eq:Def_del_c}
 \del_c \Phi(t_1, \dots, t_k) \doteq \Phi^{(1)}(\bar \ud c, t_1, \dots, t_k),
\end{equation}
which is a $k+1$-local field.
\end{remark}

\subsection{Current conservation}
\label{sec:CurrentConservation}

In \cite{HollandsWaldStress} it was shown that the principle of perturbative agreement for variations of the metric background implies the conservation of the stress-energy tensor. We want to perform the corresponding analysis for the gauge current. We define the field
%\marginpar{sign?}
\begin{equation}
\label{eq:Def_delta_j}
 (\bar \delta j)_{(SM, P, \bar X)}(c) \doteq j_{(SM, P, \bar X)}(\bar \ud c), \quad c \in \Gamma_c^\infty(M, E^0). 
\end{equation}
It is easily checked that as an element of $\F(SM, P, \bar X)$, this vanishes on all on-shell configurations. As the grade of this functional is 2, it follows that its image under $\TO$, i.e., $\alpha_H$, \cf \eqref{eq:F_to_A}, is a c-number. The principle of perturbative agreement implies a Ward identity for $\bar \delta j$:

\begin{proposition}
\label{prop:WI}
If $G$ is semisimple and \eqref{eq:PPA_gauge} holds, then for arbitrary fields $\Phi_i$,
\begin{multline}
\label{eq:WardIdentity}
 i \hbar^{-1} \TO(\Phi_1(t_1), \dots, \Phi_k(t_k), (\bar \delta j)(c)) = \\ \sum_j \TO(\Phi_1(t_1), \dots, (\Phi_j(\Lie_c t_j) - \Phi_j^{(1)}(\bar \ud c, t_j)), \dots \Phi_k(t_k)).
\end{multline}
This also holds if $G$ is not semisimple but
\begin{equation}
\label{eq:delta_j}
 \TO((\bar \delta j)(c)) = 0,
\end{equation}
for all $c \in \Gamma^\infty_c(M, E^0)$.
%Here the functional derivative is \wrt the configuration and, for $B_k = u_1 \wedge \dots \wedge u_k$,
%\[
% F^{(1)}(c \wedge \cdot)(B_k) = F^{(1)}((c \wedge u_1) \wedge u_2 \wedge \dots u_k),
%\]
%where antisymmetrization is implied.
\end{proposition}
\begin{proof}
The proof closely follows the proof of \cite[Thm~5.1]{HollandsWaldStress}. We begin by showing \eqref{eq:delta_j} for semisimple $G$.
We consider an infinitesimal gauge transformation $A = \bar \ud c'$, with $c' \in \Gamma_c^\infty(M,E^0)$. It follows from Lemma~\ref{lemma:deltaMorphism}, Remark~\ref{rem:GaugeTrafo} and the fact that $\TO$ applied to a field yields a quantum field that
\[
 \delta_\ret^{\bar \ud c'} \TO((\widetilde{\bar \delta j})(c)) = \TO(( \bar \delta j)(c' \wedge c)).
\]
%The infinitesimal gauge transformation 
%\[
% \delta_\ret \TO((\bar \delta j))(\bar \ud c', c) = - \TO(( \bar \delta j)([c',c])).
%\]
On the other hand, from \eqref{eq:PPA_gauge}, we have
\[
 \delta_\ret^{\bar \ud c'} \TO((\widetilde{\bar \delta j})(c)) = i \hbar^{-1} \Ret((\bar \delta j)(c); (\bar \delta j)(c')) + \TO((\bar \delta j)^{(1)}(\bar \ud c', c))
\]
%\[
% \delta_\ret \TO((\bar \delta j))(\bar \ud c', c) \sim i \hbar^{-1} \Ret((\bar \delta j)(c); (\bar \delta j)(c')).
%\]
We now interchange the role of $c$ and $c'$ and subtract the resulting identity. We obtain
\begin{multline}
\label{eq:nabla_j_aux}
 2 \TO(( \bar \delta j)( c' \wedge c)) = i \hbar^{-1} \Ret((\bar \delta j)(c); (\bar \delta j)(c')) - i \hbar^{-1} \Ret((\bar \delta j)(c'); (\bar \delta j)(c)) \\
 + \TO((\bar \delta j)^{(1)}(\bar \ud c', c)) - \TO((\bar \delta j)^{(1)}(\bar \ud c, c'))
\end{multline}
%\begin{equation}
%\label{eq:nabla_j_aux}
% - 2 \TO(( \bar \delta j)([c',c])) \sim i \hbar^{-1} \Ret((\bar \delta j)(c); (\bar \delta j)(c')) - i \hbar^{-1} \Ret((\bar \delta j)(c'); (\bar \delta j)(c)).
%\end{equation}
%\begin{multline}
%\label{eq:nabla_j_aux}
% - 2 \TO(( \bar \delta j)([c',c])) \sim i \Ret((\bar \delta j)(c); (\nabla j)(c')) - i \Ret((\bar \delta j)(c'); (\bar \delta j)(c)) \\ + \TO((\bar \delta j)^{(1)}(\bar \ud c', c)) - \TO((\bar \delta j)^{(1)}(\bar \ud c, c')).
%\end{multline}
For the first two terms on the r.h.s., we get
\[
 i \hbar^{-1} [\TO((\bar \delta j)(c)), \TO((\bar \delta j)(c'))]_{\star}.
\]
This is a commutator of c-numbers, so it vanishes.
%For the remaining two-terms on the \rhs of \eqref{eq:nabla_j_aux}, we obtain, using \eqref{eq:Def_delta_j}, \eqref{eq:Def_j} and the symmetry of the second order functional derivative,
%\[
% \TO(j(\bar \ud c' \wedge c)) - \TO(j(\bar \ud c \wedge c')) = \TO(j(\bar \ud (c \wedge c'))) = \TO((\bar \delta j)([c,c'])).
%\]
The last two terms in \eqref{eq:nabla_j_aux} yield
\[
 \TO(\del_{c'} \del_c S(f) - \del_c \del_{c'} S(f)) = \TO(\del_{c' \wedge c} S(f)),
\]
%\[
% \del_{c'} \del_c S_0(f) - \del_c \del_{c'} S_0(f) = \del_{[c', c]} S_0(f),
%\]
where we choose $f$ to be equal to one on the intersection of the supports of $c$ and $c'$. We used the notation \eqref{eq:Def_del_c} and computed, for some functional $F$ of the background connection, in local coordinates,
\begin{align*}
 \del_{c'} \del_{c} F[\bar A_\mu] & = \del_{c'} \tfrac{\ud}{\ud t} F[\bar A_\mu + t \del_\mu c + t [\bar A_\mu, c]]|_{t=0} \\
 & = \tfrac{\ud}{\ud s} \tfrac{\ud}{\ud t} F[\bar A_\mu  + s \del_\mu c' + s [\bar A_\mu, c'] + t \del_\mu c + t [\bar A_\mu + s \del_\mu c' + s [A_\mu,c'], c]]|_{s=t=0}.
\end{align*}
Taking the difference with $\del_{c} \del_{c'} F[\bar A]$ yields $[\del_{c'}, \del_{c}] = \del_{c' \wedge c}$. Hence, \eqref{eq:nabla_j_aux} gives
\[
 \TO((\bar \delta j)(c' \wedge c)) = 0.
\]
As by assumption $\g$ is semisimple, \eqref{eq:delta_j} follows.

It remains to prove \eqref{eq:WardIdentity}. Using Lemma~\ref{lemma:deltaMorphism}, Remark~\ref{rem:GaugeTrafo}, and \eqref{eq:PPA_gauge} for $A = \bar \ud c$, we obtain
\begin{multline*}
 \sum_j \TO(\Phi_1(t_1), \dots, (\Phi_j(\Lie_c t_j) - \Phi_j^{(1)}(\bar \ud c, t_j)), \dots \Phi_k(t_k)) \\ = i \hbar^{-1} \Ret(\Phi_1(t_1), \dots, \Phi_k(t_k); (\bar \delta j)(c)),
\end{multline*}
%\begin{multline*}
% i \Ret(\Phi_1(t_1), \dots, \Phi_k(t_k); (\bar \delta j)(c)) \\ + \sum_j \TO(\Phi_1(t_1), \dots, \Phi_j(\Lie_c t_j), \dots \Phi_k(t_k)) \sim 0,
%\end{multline*}
%\begin{multline*}
% i \Ret(\Phi_1(t_1), \dots, \Phi_k(t_k); (\bar \delta j)(c)) \\ + \sum_j \TO(\Phi_1(t_1), \dots, \Phi_j(\Lie_c t_j) + \Phi_j^{(1)}(\bar \ud c, t_j), \dots \Phi_k(t_k)) \sim 0.
%\end{multline*}
%For a field $\Phi$, we have
%\[
% \Phi(\Lie_c t) + \Phi^{(1)}(\bar \ud c, t) + \Phi(t)^{(1)}(c \wedge \cdot) = 0,
%\]
so that, using
\begin{multline}
\label{eq:RetIdentity}
 \Ret(\Phi_1(t_1), \dots, \Phi_k(t_k); \Psi(t)) = \TO (\Phi_1(t_1), \dots, \Phi_k(t_k), \Psi(t)) \\ - \TO(\Psi(t)) \star \TO (\Phi_1(t_1), \dots, \Phi_k(t_k))
\end{multline}
and \eqref{eq:delta_j}, we obtain \eqref{eq:WardIdentity}.
\end{proof}

\begin{remark}
We may define the gauge transformation on a functional $F$ by
\[
 (\Lie_c F)(B) \doteq F(\Lie_c B),
\]
where, as in Remark~\ref{rem:GaugeTrafo}, $\Lie_c B = \rho(c) B$.
For a field $\Phi$, we have
\[
 \Phi(\Lie_c t) + \Lie_c (\Phi(t)) - \Phi^{(1)}(\bar \ud c, t) = 0.
\]
Hence, we may write \eqref{eq:WardIdentity} in a form more similar to \cite[Thm.~5.1]{HollandsWaldStress} as
\begin{multline*}
%\label{eq:WardIdentity}
 i \hbar^{-1} \TO(\Phi_1(t_1), \dots, \Phi_k(t_k), (\bar \delta j)(c)) = \\ - \sum_j \TO(\Phi_1(t_1), \dots, \Lie_c (\Phi_j(t_j)), \dots \Phi_k(t_k)).
\end{multline*}
%As a consequence of Remark~\ref{rem:FieldMorphism}, we have
%\[
% \sum_j \Phi_{(SM, P, \bar X)}(t_1, \dots, \Lie_c t_j, \dots, t_k)(B) = - \Phi_{(SM, P, \bar X)}(t_1, \dots, t_k)(\Lie_c B).
%\]
%Using this, \eqref{eq:WardIdentity} may be rewritten in a form more similar to \cite[Thm.~5.1]{HollandsWaldStress}.
\end{remark}

As shown in \cite{LocCovDirac}, the parametrix can be chosen such that \eqref{eq:delta_j} holds.
Then the remaining ambiguity in the parametrix consists of changes that modify $j$ by adding a locally and covariantly constructed covector $j'$ that is conserved, $\bar \delta j' = 0$. The only such covector is the current $\bar j$ responsible for the background field. In particular, the current is unique in the absence of background currents.

Analogously to the metric variations \cite[Thm.~5.3]{HollandsWaldStress}, there is also a Ward identity for the interacting field. One defines the current as
\[
 j^\ia_{(SM, P, \bar X)}(A) \doteq (S+S^\ia)^{(1)}(A, f)
\]
with $A \in \Gamma_c^\infty(M, E^1)$ and $f=1$ on $\supp A$. Its divergence $\bar \delta j^\ia$ is defined as in \eqref{eq:Def_delta_j}. Then we have:
\begin{proposition}
If \eqref{eq:PPA_gauge} and \eqref{eq:delta_j} holds, then for arbitrary fields $\Phi_i$,
\begin{multline}
\label{eq:WardIdentity_Interacting}
 i \hbar^{-1} \TO_{S^\ia}(\Phi_1(t_1), \dots, \Phi_k(t_k), (\bar \delta j^\ia)(c)) = \\
- \sum_j \TO_{S^\ia}(\Phi_1(t_1), \dots, \Lie_c \Phi_j(t_j), \dots \Phi_k(t_k)).
%\sum_j \TO_{S^\ia}(\Phi_1(t_1), \dots, (\Phi_j(\Lie_c t_j) - \Phi_j^{(1)}(\bar \ud c, t_j)), \dots \Phi_k(t_k)).
\end{multline}
\end{proposition}

\begin{proof}
We proceed as in \cite[Thm.~5.3]{HollandsWaldStress}. We first show the result for the case $k=0$. However, to facililate the proof for other $k$, we allow for localized sources in the interaction Lagrangean, i.e., we allow for
\[
 \hat S^\ia(f, t_1, \dots, t_k) = S^{\ia}(f) + \Phi_1(t_1) + \dots + \Phi_k(t_k),
\]
where $S^\ia$ is as before and the test tensors $t_i$ are arbitrary. Note that $\hat S^\ia$ is a sum of fields. %By Remark~\ref{rem:Phi_1_Field}, the corresponding current $\hat \jmath^\ia$ is a field.
Let us for simplicity assume that $S^\ia$ and the $\Phi_i$ have even grade. The changes necessary to account for generic grades are straightforward.
By Proposition~\ref{prop:WI}, we have
\begin{multline*}
%\label{eq:WardIdentity_Interacting}
 i \hbar^{-1} \TO(\underbrace{\hat S^\ia(f, t), \dots, \hat S^\ia(f, t)}_{l \text{ times}}, (\bar \delta j)(c))  = \\ - l \TO \left( \left( (\bar \delta j^\ia)(c) - (\bar \delta j)(c) + \sum_i \Lie_c \Phi_i(t_i) \right), \underbrace{\hat S^\ia(f, t), \dots, \hat S^\ia(f, t)}_{l-1 \text{ times}} \right),
\end{multline*}
where $\hat S^\ia(f, t)$ stands for $\hat S^\ia(f, t_1, \dots, t_k)$. It follows that
\begin{multline*}
 \sum_{l=0}^\infty \frac{i^l}{l! \hbar^l} \TO( (\bar \delta j^\ia)(c), \underbrace{\hat S^\ia(f, t), \dots, \hat S^\ia(f, t)}_{l \text{ times}}) \\
 = - \sum_{l=0}^\infty \frac{i^l}{l! \hbar^l} \sum_i \TO( \Lie_c \Phi_i(t_i), \underbrace{\hat S^\ia(f, t), \dots, \hat S^\ia(f, t)}_{l \text{ times}}).
\end{multline*}
With $k = 0$, this proves \eqref{eq:WardIdentity_Interacting} for $k=0$, namely the vanishing of the divergence of the interacting current. The general result follows from replacing $\Phi_i(t_i)$ by $\lambda_i \Phi_i(t_i)$ and evaluating the derivatives \wrt all $\lambda_i$ at zero.
\end{proof}

\subsection{The background field method}
\label{sec:BFM}

Let us now discuss an application of perturbative agreement, namely the background field method. In \cite{LocCovDirac} it was already used to compute the fermion contribution\footnote{By this we mean, in the usual Feynman graph notation, that all internal lines are fermions.} to the renormalization group flow at the one-loop level. Concretely, we split the gauge connection into a background $\bar A$ and a perturbation $A$. The Dirac Lagrangean at first order in $A$ is just $j(A)$, which is now seen as an interaction term for $A$. However, as we are only interested in the fermion contribution, we may consider $A$ as parameter, instead of a field. We compute the renormalization group flow of this term by the scaling behavior, \cf \cite{HollandsWaldRG}. This gives the renormalization group flow at first order in $A$, which was computed to be \cite{LocCovDirac}
%proportional to $A_\mu \bar \nabla_\nu \bar F^{\mu\nu}$.
\begin{equation}
\label{eq:CurrentFlow}
 r_\lambda(j(A)) \propto \hbar \int \skal{\bar \ud A}{\bar F}(x) \ud_{\bar g} x.
\end{equation}
Here $\bar F$ is the curvature of the background connection, and the pairing on sections of $E^k$ is fiberwisely defined by
\[
 \skal{[p, \xi] \otimes \omega}{[p, \eta] \otimes \nu} = \kappa(\xi, \eta) \skal{\omega}{\nu}_{\bar g}, \qquad p \in SM+P, \xi, \eta \in \g, \omega, \nu \in \Omega_{\pi(p)}^k.
\]
Here $\kappa$ is the Killing form and $\skal{\cdot}{\cdot}_{\bar g}$ is the pairing of forms induced by the metric $\bar g$. Assuming that the flow does not depend on the splitting into $\bar A$ and $A$, i.e., only depends on $\bar A + A$, this gives the renormalization group flow to all orders of $A$, namely the Yang--Mills action. It remains to show that this assumption is justified. Not surprisingly, it is a consequence of perturbative agreement. To see this, we recall \cite{HollandsWaldRG} that the renormalization group flow of the S matrix for an interaction term $S^\ia$ is determined by $r_\lambda(e^{i S^\ia/\hbar})$,
%the local part of the scaling\footnote{By the local part of the scaling of a $k$-local product, we mean the term where $I_1 = \{ 1, \dots, k \}$ in the notation of \eqref{eq:Scaling}.} of the S matrix $\TO(e^{ij(A)/\hbar})$, i.e., by $r(e^{ij(A)/\hbar})$,
%\[
%  \sum_{k=0}^\infty \frac{i^k}{k! \hbar^k} \TO(r(j(A)^k)),
%\]
where $r_\lambda$ are the renormalization maps appearing in \eqref{eq:Scaling}. We now have the following:
%As the scaled time-ordered products fulfill the requirements of time-ordered products, the $r$ fulfill the properties discussed in Remark~\ref{rem:Renormalization}. In particular, they vanish if one of the arguments is a linear field and fulfill {\bf Expansion}. As $j(A)$ is quadratic in the Dirac fields, it follows that the functional derivative of $r(j(A)^k)$ vanishes, i.e., it is a c-number. 
%Let us now discuss what happens if we change that split into background and interaction, i.e., if we perform $\bar A \to \bar A + A'$ and $A \to A - A'$. Assuming perturbative expansion \eqref{eq:PPA_gauge}, we have
%\[
% \delta_\ret^{A'} \TO(\tilde \jmath(A)^k) = i \hbar^{-1} \Ret(j(A)^k; j(A')),
%\]
%where we used that $j(A)$ is independent of the gauge background. Using \eqref{eq:delta_sigma}, we obtain
%\[
% \delta_\ret^{A'} \sigma_\lambda \TO(\tilde \jmath(A)^k) = i \hbar^{-1} \sigma_\lambda \Ret(j(A)^k; j(\lambda A')).
%\]
%The retarded product on the \rhs is a multilocal field, and the local part of its scaling is given by $r(j(A)^k, j(A'))$, as the second term on the \rhs of \eqref{eq:RetIdentity} does not contribute to the local 
\begin{proposition}
$r_\lambda(j(A)^k)$ is a c-number field, which, if \eqref{eq:PPA_gauge} is fulfilled, satisfies
\begin{equation}
\label{eq:r_A}
 r_\lambda(e^{i j(A) / \hbar})^{(1)}(A') = i \hbar^{-1} \left(  r_\lambda(e^{i j(A) / \hbar}, j(A')) - r_\lambda(j(A')) \right).
\end{equation}
%\[
% \delta_\ret^{A'} \TO(r(\tilde \jmath(A)^k)) = i \hbar^{-1} \TO(r(j(A)^k, j(A')),
%\]
%where $r$ is the field defined in \eqref{eq:Scaling}.
\end{proposition}
\begin{proof}
To see that $r_\lambda(j(A)^k)$ is a c-number field, recall that $r_\lambda$ fulfills the requirements discussed in Remark~\ref{rem:Renormalization}. In particular, it vanishes if one of the arguments is a linear field and fulfills {\bf Expansion}. As $j(A)$ is quadratic, it follows that the functional derivative of $r_\lambda(j(A)^k)$ vanishes, i.e., it is a c-number.

In order to prove the second statement, we note that, as a consequence of the scaling formula \eqref{eq:Scaling}, we have, with $A_\lambda \doteq \lambda^{n} A$,
\begin{equation}
\label{eq:Scale_T_j}
 S_\lambda \TO \left( e^{i j(A) / \hbar} \right) = \TO \left( \exp \left( i j(A_\lambda) / \hbar + r_\lambda(e^{i j(A_\lambda) / \hbar}) \right) \right).
\end{equation}
Furthermore, by \eqref{eq:PPA_gauge},
\[
 \delta_\ret^{A'} \TO \left( e^{i \tilde \jmath(A) / \hbar } \right) = i \hbar^{-1} \Ret \left( \exp \left( i j(A) / \hbar \right); j(A') \right),
\]
where we used that the $j(A)$ are independent of the background connection. Using \eqref{eq:RetIdentity} and
\[
 \TO \left( e^{ i j(A) / \hbar }, j(A') \right) = - i \hbar \frac{\ud}{\ud s} \TO \left( e^{ i j(A + s A') / \hbar } \right) |_{s=0},
\]
we obtain
\begin{multline*}
 S_\lambda \Ret \left( e^{ i j(A) / \hbar }; j(A') \right) = \TO \left( \exp \left( i j(A_\lambda) / \hbar + r_\lambda(e^{i j(A_\lambda) / \hbar}) \right), j(A'_\lambda) \right) \\
 + \TO \left( \exp \left( i j(A_\lambda) / \hbar + r_\lambda(e^{i j(A_\lambda) / \hbar}) \right), r_\lambda(e^{i j(A_\lambda) / \hbar}, j(A'_\lambda)) \right) \\
 - \left( \TO(j(A'_\lambda)) + r_\lambda(j(A'_\lambda)) \right) \star  \TO \left( \exp \left( i j(A_\lambda) / \hbar + r_\lambda(e^{i j(A_\lambda) / \hbar}) \right) \right),
\end{multline*}
and hence
\begin{multline}
\label{eq:S_delta}
 -i \hbar S_\lambda \delta_\ret^{A'} \TO \left( e^{i \tilde \jmath(A) / \hbar } \right) = \Ret \left( \exp \left( i j(A_\lambda) / \hbar + r_\lambda(e^{i j(A_\lambda) / \hbar}) \right), j(A'_\lambda) \right) \\
 + \left(  r_\lambda(e^{i j(A_\lambda) / \hbar}, j(A'_\lambda)) - r_\lambda(j(A'_\lambda)) \right) \TO \left( \exp \left( i j(A_\lambda) / \hbar + r_\lambda(e^{i j(A_\lambda) / \hbar}) \right) \right).
\end{multline}
On the other hand, due to \eqref{eq:PPA_gauge} and \eqref{eq:Scale_T_j}, we have
\begin{multline}
\label{eq:delta_S}
 \delta_\ret^{A'_\lambda} S_\lambda \TO \left( e^{i \tilde \jmath(A) / \hbar} \right) =
i \hbar^{-1} \Ret \left( \exp \left( i j(A_\lambda) / \hbar + r_\lambda(e^{i j(A_\lambda) / \hbar}) \right), j(A'_\lambda) \right) \\
+  r_\lambda(e^{i j(A_\lambda) / \hbar})^{(1)}(A'_\lambda) \TO \left( \exp \left( i j(A_\lambda) / \hbar + r_\lambda(e^{i j(A_\lambda) / \hbar}) \right) \right).
\end{multline}
It follows from \eqref{eq:ScaledField} and \eqref{eq:delta_sigma} that
\[
 \delta_\ret^{A_\lambda} S_\lambda = S_\lambda \delta_\ret^A,
\]
which proves \eqref{eq:r_A} upon comparison of \eqref{eq:S_delta} and \eqref{eq:delta_S}.
%Now let us proceed by induction. For $k=1,2$, we have
%\begin{align}
%\label{eq:r_j_1}
% \lambda^{1-n} S_\lambda \TO(j(A)) & = \TO(j(A)) + r(j(A)) \\
%\label{eq:r_j_2}
% \lambda^{2-2n} S_\lambda \TO(j(A)^2) & = \TO(j(A)^2) + 2 \TO(j(A)) r(j(A)) + r(j(A)^2),
%\end{align}
%where we used that the $r$'s are central. On the other hand, we have
%\begin{align*}
% \delta_\ret^{A'} \lambda^{1-n} S_\lambda \TO( \tilde \jmath(A)) & = \lambda \delta_\ret^{A'} \sigma_\lambda \TO( \tilde \jmath(A)) \\
% & = \lambda \sigma_\lambda \delta_\ret^{\lambda A'} \TO( \tilde \jmath(A)) \\
% & = i \hbar^{-1} \lambda^{2-2n} S_\lambda \Ret(j(A); j(A')).
%\end{align*}
%In the last step, we used \eqref{eq:PPA_gauge} and the fact that the currents are independent of the background connection. Using \eqref{eq:r_j_2} and \eqref{eq:RetIdentity}, we obtain
%\[
% \lambda^{2-2n} S_\lambda \Ret(j(A); j(A')) = \Ret(j(A); j(A')) + r(j(A), j(A')).
%\]
%Comparison with \eqref{eq:r_j_1} yields the result for $k=1$.
\end{proof}

Using the definition \eqref{eq:DefFieldDerivative} of the derivative of a field \wrt the background geometry, we may write \eqref{eq:r_A} as
\begin{equation*}
%\label{eq:r_A}
 \tfrac{\ud}{\ud s} i_{\bar X, \bar X - s A'} \tilde r_\lambda(e^{i j(A + s A') / \hbar})|_{s=0} = i \hbar^{-1} r_\lambda(j(A')).
\end{equation*}
Hence, we may compute the renormalization group flow of the interaction term $j(A)$ due to the fermions as
\begin{align*}
 r_{\lambda, \bar A}(e^{i j(A) / \hbar}) & = \int_0^1 \tfrac{\ud}{\ud t} i_{\bar X, \bar X + (1-t) A} \tilde r_{\lambda, \bar A+(1-t)A}(e^{i j(t A) / \hbar})|_{t=s} \ud s \\
 & = i \int_0^1 \skal{(\bar \ud + s A \wedge) A}{\bar F + s \bar \ud A + \tfrac{s^2}{2} A \wedge A} \ud s,
\end{align*}
where we used \eqref{eq:CurrentFlow} and $F[\bar A + A] = \bar F + \bar \ud A + \frac{1}{2} A \wedge A$. Carrying out the integration, we obtain $\frac{1}{2} (\skal{F[\bar A+A]}{F[\bar A+A]} - \skal{\bar F}{\bar F})$, i.e., the Yang-Mills action, up to an $A$-independent term. Hence, one can infer the fermion contribution to the renormalization group flow from the flow at first order in $A$. We note that the latter is given by a coinciding point limit of Hadamard coefficients \cite{LocCovDirac}. These also appear in the heat kernel expansion (as Seeley-de Witt coefficients), which establishes a relation to the heat kernel method.
%See \cite{HackMoretti} for related discussions.

%We may write \eqref{eq:r_A} as
%\begin{equation*}
%%\label{eq:r_A}
% r_\lambda(e^{i j(A) / \hbar})^{(1)}(A') = \tfrac{\ud}{\ud s} r_\lambda(e^{i j(A + s A') / \hbar})|_{s=0} - i \hbar^{-1} r_\lambda(j(A')).
%\end{equation*}
%It follows that up to a correction term which does not depend on $A$, the renormalization group flow is invariant under a transformation $\bar A \to \bar A + A'$, $A \to A - A'$ which only changes the splitting into background and perturbation. Hence, it is indeed legitimate the infer the fermion contribution to the renormalization group flow at $\order(\hbar)$ from the flow at first order in $A$. We note that the latter is given by a coinciding point limit of Hadamard coefficients \cite{LocCovDirac}. These also appear in the heat kernel expansion (as Seeley-de Witt coefficients), which establishes a relation to the heat kernel method.
%%See \cite{HackMoretti} for related discussions.

\subsection{Fulfillment of perturbative agreement}
\label{sec:FulfillmentPPA}
We now want to investigate whether the principle of perturbative agreement can be satisfied. We proceed analogously to \cite{HollandsWaldStress}. There, one defines the deviation
\begin{align}
\label{eq:defD}
 D_k(A, t_1, \dots, t_k) & \doteq \delta_\ret^A \TO(\tilde \Phi_1(t_1), \dots, \tilde \Phi_k(t_k)) \\
 & - i \hbar^{-1} \Ret(\Phi_1(t_1), \dots, \Phi_k(t_k); j(A)) \nonumber \\
 & - \sum_j \TO(\Phi_1(t_1), \dots, \Phi_j^{(1)}(A, t_j), \dots, \Phi_k(t_k)). \nonumber
\end{align}
%\begin{align*}
% D_k(A, t_1, \dots, t_k) & \doteq  (\delta_\ret \TO(\Phi_1, \dots, \Phi_k))(A, t_1, \dots t_k) \\
% & - i \Ret(\Phi_1(t_1), \dots, \Phi_k(t_k); j(A)),
%\end{align*}
and then proceeds inductively in the total grade $N$ of the fields $\Phi_i$. For $N=0$, only $k=0$ is nontrivial, but $D_0 = 0$ identically. One assumes that $D_k = 0$ for all fields with total grade smaller than $N$. The idea is then to modify the time-ordered product $\TO(\Phi_1(t_1), \dots, \Phi_k(t_k), j(A))$ appearing in the second term on the \rhs such that the modified $D_k$ vanishes. Such a modification is only possible if the $D_k$ fulfill a couple of properties, \cf the discussion of the renormalization freedom in Remark~\ref{rem:Renormalization}. We proceed in close analogy to \cite{HollandsWaldStress}, only highlighting the differences.
%Then one can show that the $D_k$ for total grade $N$ are numerical distributions, supported on the diagonal $\diag_{k+1}$, which are locally and covariantly constructed and scale homogeneously with the correct scaling dimension. The idea is then to modify the the c-number part of the time-ordered product $\TO(\Phi_1(t_1), \dots, \Phi_k(t_k), j(A))$ such that the modified $D_k$ vanishes. In order for this 

As in \cite{HollandsWaldStress}, one shows that $D_k$ is a c-number supported on the total diagonal $\diag_{k+1}$. It is local and covariant, scales almost homogeneously, and vanishes if one of the $\Phi_i$ is a linear field. We also have
\[
 \overline{D_k(A, t_1, \dots, t_k)} = (-1)^{k+1} D_k^*(A, t_k^*, \dots, t_1^*),
\]
where $D_k^*(A, t_k^*, \dots, t_1^*)$ is defined as in \eqref{eq:defD}, but with $\Phi_1(t_1), \dots, \Phi_k(t_k)$ replaced by $\Phi_k(t_k)^*, \dots, \Phi_1(t_1)^*$. This follows from the commutation of $\delta^A_\ret$ with the adjoint, \cf Proposition~\ref{prop:tau_isomorphism}, {\bf Unitarity} of the time-ordered products, and the fact that for lower total grade we have $D_k = 0$.

The next step is to show that
\begin{equation}
\label{eq:D_WF_1}
 \WF(D_k)|_{\diag_{k+1}} \perp T \diag_{k+1}.
\end{equation}
Furthermore, $D_k$ depends smoothly and (if applicable) analytically on the background fields, i.e., for a smooth (analytic) family $\mathcal{P} \ni s \mapsto \bar X_s$ of backgrounds, with $\mathcal{P}$ a finite dimensional parameter space, we have
\begin{equation}
\label{eq:D_WF_2}
 \WF(D^s_k)|_{\mathcal{P} \times \diag_{k+1}} \perp T (\mathcal{P} \times \diag_{k+1}).
\end{equation}
For the second and third term in $D_k$, this follows from the properties of the retarded product.
For the first term, one restricts to a sufficiently small neighborhood of $\diag_{k+1}$, where one can write
\begin{multline*}
 \TO(\Phi_1(t_1), \dots, \Phi_k(t_k)) = \\ \int w(y_1, \dots, y_k; x_1, \dots x_N) \WDp{\Psi(x_1) \dots \Psi(x_N)}_H \prod_i \ud_{\bar g} x_i \prod_j t(y_j) \ud_{\bar g} y_j,
\end{multline*}
where $\WDp{\cdot}_H$ denotes Hadamard normal ordering \cite{HollandsWaldStress, LocCovDirac}, and $w$ is a distributional section. $\Psi$ stands for either $\psi$ or $\psi^+$ in the notation used in \eqref{eq:CurrentSuccinct}. The dependence on the background resides in $w$ and the Hadamard normal ordered product. The variation \wrt the background thus gives two terms, one from the variation of $w$ and one from the variation of the normal ordered product.\footnote{For metric variations, also the volume elements have to be varied, leading to $\delta(y_j, y)$ or $\delta(x_i,y)$ above. Knowing the wave front set of $w$ and the Hadamard ordered product, this fulfills \eqref{eq:D_WF_1} and \eqref{eq:D_WF_2}.}
Both fulfill the conditions \eqref{eq:D_WF_1} and \eqref{eq:D_WF_2} independently.
%For the treatment of the variation of $w$, we refer to \cite[Section 6.2.5]{HollandsWaldStress}. The variation of the normal ordered product is also treated there, but the proof contains a few mistakes. 
Let us first discuss the variation of the normal ordered product. As the treatment in \cite[Section 6.2.5]{HollandsWaldStress} contains a few mistakes, we sketch how to correct these and how to implement the changes necessary to treat Dirac spinors. For the wave front set of $w$, one knows from the {\bf Microlocal spectrum condition} that
\begin{multline}
\label{eq:WF_s}
 \WF(w)|_{\diag_k} \subset \\ \{ (y,p_1; \dots y,p_k; x_1, k_1; \dots x_N, k_N) | x_i = y \ \forall i, \sum p_i + \sum k_j = 0 \}.
\end{multline}
In order to treat the variation of the Hadamard normal ordered product, one evaluates it in a suitable family $\omega_s$ of quasi-free Hadamard states (whose two-point functions will be denoted by the same symbol). Pick a Hadamard two-point function $\omega$ for $s=0$ and define the corresponding family $\omega_s$ as in Definition~\ref{def:tau}. Define
\[
 d_s \doteq \omega_s - H_s,
\]
where $H_s$ is the parametrix for the background $\bar X_s$. Then
\[
 \omega \left( \tau_\ret^{\bar X, \bar X_s} \WDp{\Psi(x_1) \dots \Psi(x_N)}_{H_s} \right) = \prod_{\text{pairs } ij} d_s(x_i, x_j),
\]
so that the part of the variation of $\omega(\TO(\Phi_1(t_1), \dots, \Phi_k(t_k)))$ that comes from the variation of the Hadamard ordered products can be written as
\[
 \int w_0(y_1, \dots, y_k; x_1, \dots x_r) \del_s \prod_{\text{pairs } ij} d_s(x_i, x_j)|_{s=0} \prod_i \ud_{\bar g} x_i \prod_j t(y_j) \ud_{\bar g} y_j.
\]
If $X$ is the infinitesimal variation corresponding to $\bar X_s$, then one defines
\[
 (\delta d)(X, u_1, u_2) \doteq \del_s d_s(u_1, u_2)|_{s=0}.
\]
According to the above, the wave front set of $\delta d$ at the diagonal is crucial for the determination of the wave front set of $D_k$. It is characterized by the following lemma, which corresponds to \cite[Lemma~6.2]{HollandsWaldStress}:
\begin{lemma}
\label{lemma:Smoothness}
For a sufficiently small causal domain $U \subset M$, with the support $K$ of the background variation contained in $U$, $d_s(x_1, x_2)$ is jointly smooth in $(s, x_1, x_2)$ for $x_i \in U$. Furthermore,
\begin{equation}
\label{eq:WF_delta_d}
 \WF(\delta d) |_{\diag_3} \subset \{ (y,p; x_1, k_1; x_2, k_2) | y=x_1=x_2, p+k_1+k_2 = 0  \}.
\end{equation}
\end{lemma}
A causal domain is a globally hyperbolic open subset, whose closure is contained in a geodesically convex region.
As the proof given in \cite{HollandsWaldStress} contains a mistake, we provide a corrected proof in Appendix~\ref{app:Smoothness}. It also includes the changes that are necessary to treat the Dirac equation instead of the Klein-Gordon equation. With this lemma and the knowledge about the wave front set of $w_0$, \cf \eqref{eq:WF_s}, it is then straightforward to prove \eqref{eq:D_WF_1} and \eqref{eq:D_WF_2}. For the treatment of the variation of $w$, we refer to \cite[Section 6.2.5]{HollandsWaldStress}, noting however that there equation (238) has to be corrected, analogously to the correction of equation (230) that we give in \eqref{eq:WF_delta_G}.

%\subsubsection{$D_k$ is symmetric for $\Phi_1 = j$}

The next step is to show that the $D_k$ are symmetric if one of the fields is the current, i.e., for $\Phi_1 = j$,
\begin{equation}
\label{eq:D_k_sym}
 D_k(A_1, A_2, t_2, \dots, t_k) - D_k(A_2, A_1, t_2, \dots, t_k) = 0.
\end{equation}
In fact, it will turn out that the problem can be reduced to the case $k=1$. We denote the \lhs of this equation for $k=1$ by $E(A_1, A_2)$. We have
\begin{equation}
\label{eq:Def_E}
 E(A_1, A_2) = \delta_\ret^{A_1} \TO( \tilde \jmath (A_2)) - \delta_\ret^{A_2} \TO( \tilde \jmath (A_1)) + i \hbar^{-1} [\TO(j(A_1)), \TO(j(A_2))],
\end{equation}
%\begin{equation}
%\label{eq:Def_E}
% E(A_1, A_2) = (\delta_\ret \TO(j))(A_1, A_2) - (\delta_\ret \TO(j))(A_2, A_1) + i [\TO(j(A_1)). \TO(j(A_2))],
%\end{equation}
where we used that $S^{(2)}(A_1, A_2, f)$ is symmetric in $A_1$, $A_2$ (in fact, for the Dirac field $S$ is linear in the connection, so the second derivative vanishes anyway).

\begin{proposition}
If $E$ vanishes and $D_j$ vanishes for all $j$ and all total grades $N\leq \sum_{i=2}^k \betrag{\Phi_i}$, then \eqref{eq:D_k_sym} is fulfilled for all $k$.
\end{proposition}

\begin{proof}
The statement is only nontrivial for $k \geq 2$. In that case, we have
%We proceed by induction in $k$. If $E$ vanishes, then \eqref{eq:D_k_sym} is fulfilled for $k=2$, and hence we have an induction start. Now assume that \eqref{eq:D_k_sym} is fulfilled up to $k-1$. We have
\begin{align}
\label{eq:Line_1}
 \text{\eqref{eq:D_k_sym}} & = \delta_\ret^{A_1} \TO( \tilde \jmath (A_2), \tilde \Phi_2(t_2), \dots, \tilde \Phi_k(t_k)) - (1 \leftrightarrow 2) \\
\label{eq:Line_2}
 & - i \hbar^{-1} \Ret(j(A_2), \Phi_2(t_2), \dots, \Phi_k(t_k); j(A_1)) - (1 \leftrightarrow 2) \\
\label{eq:Line_3}
 & - \sum_{j=2}^k \TO(j(A_2), \Phi_2(t_2), \dots, \Phi_j^{(1)}(A_1 * t_j), \dots, \Phi_k(t_k)) - (1 \leftrightarrow 2),
\end{align}
where we again used that $S^{(2)}(A_1, A_2, f)$ is symmetric in $A_1, A_2$. We may replace the second line by
\begin{equation}
\label{eq:Line_2_1}
 \text{\eqref{eq:Line_2}} = i \hbar^{-1} \TO(j(A_1)) \star \TO(j(A_2), \Phi_2(t_2), \dots, \Phi_k(t_k)) - (1 \leftrightarrow 2).
\end{equation}
The first line may be rewritten as
\begin{align}
\label{eq:Line_1_1}
 \text{ \eqref{eq:Line_1} } & = \delta_\ret^{A_1} \Ret(\tilde \Phi_2(t_2), \dots, \tilde \Phi_k(t_k); \tilde \jmath(A_2)) - (1 \leftrightarrow 2) \\
\label{eq:Line_1_2}
 & + \delta_\ret^{A_1} (\TO(\tilde \jmath(A_2)) \star \TO(\tilde \Phi_2(t_k), \dots, \tilde \Phi_k(t_k))) - (1 \leftrightarrow 2).
\end{align}
Due to the assumption on the vanishing of $D_j$, \eqref{eq:Line_1_1} can be rewritten as
\begin{align}
\label{eq:Line_1_1_1}
 \text{ \eqref{eq:Line_1_1} } & = -i \hbar \delta_\ret^{A_1} \delta_\ret^{A_2} \TO(\tilde{\tilde \Phi}_2(t_2), \dots, \tilde{\tilde \Phi}_k(t_k)) - (1 \leftrightarrow 2) \\
\label{eq:Line_1_1_2}
  & + i \hbar \sum_{j=2}^k \delta_\ret^{A_1} \TO(\tilde \Phi_2(t_2), \dots, \tilde \Phi^{(1)}_j(A_2 * t_j), \dots, \tilde \Phi_k(t_k)) - (1 \leftrightarrow 2).
\end{align}
Here $\tilde{\tilde \Phi}$ denotes the two-parameter family corresponding to variations along $A_1$ and $A_2$. The retarded variations commute, so that the first line \eqref{eq:Line_1_1_1} vanishes. The line \eqref{eq:Line_1_2} may be expanded by the Leibniz rule \eqref{eq:delta_Leibniz} to
\begin{align}
\label{eq:Line_1_2_1}
 \text{ \eqref{eq:Line_1_2} } & = \delta_\ret^{A_1} \TO(\tilde \jmath(A_2)) \star \TO(\Phi_2(t_2), \dots, \Phi_k(t_k)) - (1 \leftrightarrow 2) \\
\label{eq:Line_1_2_2}
  & + \TO(j(A_2)) \star \delta_\ret^{A_1} \TO(\tilde \Phi_2(t_2), \dots, \tilde \Phi_k(t_k)) - (1 \leftrightarrow 2).
\end{align}
Using again the assumption on the $D_j$, \eqref{eq:Line_1_2_2} can be written as
\begin{align}
\label{eq:Line_1_2_2_1}
 \text{ \eqref{eq:Line_1_2_2} } & = i \hbar^{-1} \TO(j(A_2)) \star \Ret(\Phi_2(t_2), \dots, \Phi_k(t_k); j(A_1)) - (1 \leftrightarrow 2) \\
\label{eq:Line_1_2_2_2}
  & + \TO(j(A_2)) \star \sum_{j=2}^k \TO(\Phi_2(t_2), \dots, \Phi^{(1)}_j(A_1 * t_j), \dots \Phi_k(t_k)) - (1 \leftrightarrow 2).
\end{align}
Now \eqref{eq:Line_2_1} and  \eqref{eq:Line_1_2_2_1} add up to
\begin{equation}
\label{eq:Line_2_1_and_1_2_2_1}
 \text{\eqref{eq:Line_2_1}} + \text{\eqref{eq:Line_1_2_2_1}} = -i \hbar^{-1} [\TO(j(A_2)), \TO(j(A_1))] \star \TO(\Phi_2(t_2), \dots, \Phi_k(t_k)).
\end{equation}
Again using the vanishing of the $D_j$, \eqref{eq:Line_1_1_2} is
\begin{equation}
\label{eq:Line_1_1_2_1}
 \text{ \eqref{eq:Line_1_1_2} } = - \sum_{j=2}^k \Ret(\Phi_2(t_2), \dots, \Phi^{(1)}_j(A_2 * t_j), \dots, \Phi_k(t_k); j(A_1)) - (1 \leftrightarrow 2), %\\
%\label{eq:Line_1_1_2_2}
%  & + i \sum_{j \neq l} \Ret(\Phi_2(t_2), \dots, \Phi^{(1)}_j(A_2, t_j), \dots, \Phi^{(1)}_l(A_1, t_l), \dots, \Phi_k(t_k))- (1 \leftrightarrow 2) 
\end{equation}
where we used the commutativity of the derivative \wrt the background field. Now \eqref{eq:Line_1_1_2_1} and \eqref{eq:Line_1_2_2_2} add up to
\begin{equation*}
 \text{\eqref{eq:Line_1_1_2_1}} + \text{\eqref{eq:Line_1_2_2_2}} = - \sum_{j=2}^k \TO(j(A_1), \Phi_2(t_2), \dots, \Phi^{(1)}_j(A_2 * t_j), \dots, \Phi_k(t_k)) - (1 \leftrightarrow 2),
\end{equation*}
which cancels \eqref{eq:Line_3}. The remaining terms \eqref{eq:Line_1_2_1} and \eqref{eq:Line_2_1_and_1_2_2_1} give
\[
 \text{\eqref{eq:D_k_sym}} = E(A_1, A_2) \star \TO(\Phi_2(t_2), \dots, \Phi_k(t_k)),
\]
which vanishes by the hypothesis.
\end{proof}

Hence, it remains to show that $E=0$, which turns out to be a consequence of current conservation:

\begin{proposition}
\label{prop:E0}
For $n \leq 4$, if $\TO((\bar \delta j)(c)) = 0$, then $E=0$.
\end{proposition}

\begin{proof}
Set $A_1 = A$, $A_2 = \bar \ud c$. By Lemma~\ref{lemma:deltaMorphism}, we have
\[
 \delta_\ret^{\bar \ud c} \TO(\tilde \jmath (A)) = \TO(j(c \wedge A)).
\]
Furthermore,\footnote{Note that on the \lhs the background connection in $\bar \delta$ is not varied, contrary to the first term on the r.h.s. This is corrected for by the second term on the r.h.s.}
\[
 \delta_\ret^A \TO(\tilde \jmath(\bar \ud c)) = \delta_\ret^A \TO( (\widetilde{\bar \delta j})(c)) - \TO(j(A \wedge c)).
\]
If $\TO((\bar \delta j)(c)) = 0$, then the first term on the \rhs vanishes. Also the commutator term in \eqref{eq:Def_E} vanishes, so that
%If $T((\bar \delta j)(c)) \sim 0$, then, by Lemma~\ref{lemma:OnShell}, the first term vanishes on-shell. Also the commutator term in \eqref{eq:Def_E} vanishes on-shell, so that
\begin{equation}
\label{eq:E_A_dc}
 E(A, \bar \ud c) = 0 \qquad \forall A \in \Gamma_c^\infty(M, E^1), c \in \Gamma_c^\infty(M, E^0).
\end{equation}
On the other hand, from the above we know that $E$ is given by
\begin{equation}
\label{eq:E_A1_A2}
 E(A_1, A_2) = \sum_{r=0}^R \int {A_1}^I_\mu(x) \bar \nabla_{(\lambda_1} \dots \bar \nabla_{\lambda_r)} {A_2}^J_\nu(x) C^{\mu \nu \lambda_1 \dots \lambda_r}_{IJ}(x) \ud_{\bar g} x - (1 \leftrightarrow 2)
\end{equation}
for some locally and covariantly constructed $C^{\mu \nu \lambda_1 \dots \lambda_r}_{IJ}$ of scaling dimension $n$ and a finite $R$. Here $I, J$ are indices labelling a basis of $\g$.
%By partial integration, we can bring this to the form
%\begin{equation*}
%%\label{eq:E_A1_A2}
% E(A_1, A_2) = \sum_{r=0}^R \int {A_1}^a_\mu(x) \bar \nabla_{(\lambda_1} \dots \bar \nabla_{\lambda_r)} {A_2}^b_\nu(x) D^{\mu \nu \lambda_1 \dots \lambda_r}_{ab}(x) \ud_{\bar g} x,
%\end{equation*}
%where
%\[
% D^{\mu \nu \lambda_1 \dots \lambda_R}_{ab} = C^{\mu \nu \lambda_1 \dots \lambda_R}_{ab} - (-1)^R C^{\nu \mu \lambda_1 \dots \lambda_R}_{ba}.
%\]
Without loss of generality, we may assume that $C^{\mu \nu \lambda_1 \dots \lambda_r}_{IJ}$ is symmetric in the $\lambda_i$ and that
\begin{equation}
\label{eq:C_condition_0}
 C^{\mu \nu \lambda_1 \dots \lambda_r}_{IJ} = - (-1)^r C^{\nu \mu \lambda_1 \dots \lambda_r}_{JI},
\end{equation}
as, by starting at $r=R$, we may recursively cancel the (anti-) symmetric component for even (odd) $r$ by partial integration. And for $r=0$, the symmetric part cancels anyway. Using \eqref{eq:C_condition_0}, differentiation of \eqref{eq:E_A_dc} \wrt $A_\mu^I(x)$ yields
\begin{multline*}
 \sum_{r=0}^R \bar \nabla_{(\lambda_1} \dots \bar \nabla_{\lambda_r)} \bar \nabla_\nu c^J(x) C^{\mu \nu \lambda_1 \dots \lambda_r}_{I J}(x) \\ + \bar \nabla_{(\lambda_1} \dots \bar \nabla_{\lambda_r)} (\bar \nabla_\nu c^J (x) C^{\mu \nu \lambda_1 \dots \lambda_r}_{I J}(x)) = 0.
\end{multline*}
%\begin{multline*}
% \sum_{r=0}^R \bar \nabla_{(\lambda_1} \dots \bar \nabla_{\lambda_r)} \bar \nabla_\nu c^b(x) C^{\mu \nu \lambda_1 \dots \lambda_r}_{ab}(x) \\ - (-1)^r \bar \nabla_{(\lambda_1} \dots \bar \nabla_{\lambda_r)} (\bar \nabla_\nu c^b (x) C^{\nu \mu \lambda_1 \dots \lambda_r}_{ba}(x)) = 0.
%\end{multline*}
%Differentiation of \eqref{eq:E_A_dc} \wrt $c$ thus yields
%\begin{multline*}
% \sum_r \bar \nabla_\nu \bar \nabla_{(\lambda_1} \dots \bar \nabla_{\lambda_r)} (A_\mu(x) C^{\mu \nu \lambda_1 \dots \lambda_r}(x)) \\ - (-1)^r \bar \nabla_\nu (\bar \nabla_{(\lambda_1} \dots \bar \nabla_{\lambda_r)} A_\mu(x) C^{\mu \nu \lambda_1 \dots \lambda_r}(x)) = 0.
%\end{multline*}
We may choose $c^J$ such that $\bar \nabla_{\rho_1} \dots \bar \nabla_{\rho_l} c^J(x) = 0$ for all $l \leq R$, and the symmetric parts of $\bar \nabla_{\rho_1} \dots \bar \nabla_{\rho_{R+1}} c^J(x)$ can be chosen independently. Hence,
%\begin{equation}
%\label{eq:C_condition_1}
% C^{\mu (\nu \lambda_1 \dots \lambda_R)}_{ab} = (-1)^R C^{(\nu | \mu | \lambda_1 \dots \lambda_R)}_{ba}.
%\end{equation}
%Together with \eqref{eq:C_condition_0}, this implies
\begin{equation}
\label{eq:C_condition_1}
 C^{\mu (\nu \lambda_1 \dots \lambda_R)}_{I J} =0. % = 0 = C^{(\nu | \mu | \lambda_1 \dots \lambda_R)}_{ba}.
\end{equation}
In particular, this already excludes $R=0$.

Let us discuss the possibility to construct such tensors $C$ for dimensions $n \leq 4$.
As we have $2 + r$ upper indices, we need at least $1 + \lceil r/2 \rceil$ inverse metrics\footnote{As the model is parity even, we exclude the completely antisymmetric tensor.} ($\lceil m \rceil$ denoting the smallest number greater or equal to $m$). These have at least scaling dimension $2(1 + \lceil r/2 \rceil)$. Hence, for $n=2$, only $R = 0$ is possible, which we already excluded. The same applies to $n=3$. For $n=4$, $R=0$ is already excluded, and $R=1$ would require the existence of a covariant tensor of rank $3$ and scaling dimension $4$, which does not exist.
%For $n=2$, only $R=0$ is possible, which we already excluded. For $n-3$, we have the possibility $R=1$, but this requires the existence of a covariant tensor of rank 3 and vanishing mass dimension, which does not exist (our model is parity even, so $\eps^{\mu \nu \lambda}$ may not appear).  For $n=4$, $R=0$ is already excluded, and $R=1$ would require the existence of a covariant object of mass dimension $1$, which does not exist.
Hence, $R=2$ remains. But due to the scaling condition, only the Killing form $\kappa_{I J}$ and $g^{\mu \nu}$ can be used. Hence, the most general form of $C$ is
\[
 C^{\mu \nu \lambda_1 \lambda_2}_{I J} = \kappa_{I J} \left( c_1 g^{\mu \nu} g^{\lambda_1 \lambda_2} + c_2 g^{\mu \lambda_1} g^{\nu \lambda_2} + c_3 g^{\mu \lambda_2} g^{\nu \lambda_1} \right).
\]
Now \eqref{eq:C_condition_0} requires $c_1 = 0$ and $c_2 = -c_3$. On the other hand, symmetry in the $\lambda$'s requires $c_2 = c_3$, so $C=0$.
\end{proof}

\begin{remark}
In the case of metric variations, it is claimed in \cite[Section~6.2.6]{HollandsWaldStress} that the analogous statement is true independently of the dimension. However, the proof seems to contain a gap,\footnote{In the absence of coupling constants of negative mass dimension, one can invoke dimensional arguments to fix this gap for $n \leq 4$, analogously to the above.} related to the fact that in \eqref{eq:C_condition_1}, we can only make statements about the symmetrization in $(\nu \lambda_1 \dots \lambda_R)$ instead of symmetrization in $(\lambda_1 \dots \lambda_R)$. A thorough investigation of the ``background cohomology'' introduced in \cite{HollandsWaldStress} might be helpful in overcoming the restriction on the dimension.
\end{remark}
%\marginpar{HW05 have general $n$?! Include cohomological description?}

%\subsubsection{The time-orderded products can be redefined such that $D_k = 0$}

%Now we have all the ingredients to redefine the time-ordered product such that $D_k = 0$. Two different 

%How do we now redefine the time-ordered product? For a nontrivial finite-dimensional representation $\rho$ of a simple Lie algebra the quadratic form $\skal{\lambda}{\lambda'} = \tr_V(\rho(\lambda) \rho(\lambda'))$ is a multiple of the Killing form $\kappa$. 

It now remains to redefine the time-ordered product in order to achieve $D_k = 0$. Assume that $G$ is simple. We may also assume that $\rho$ is nontrivial, as otherwise all requirements are trivially fulfilled. In that case, the quadratic form $\skal{\lambda}{\lambda'} = \tr_V(\rho(\lambda) \rho(\lambda'))$ on $\g$ is a multiple of the Killing form $\kappa$. Hence, we may set
\[
 r^{k+1}[ \psi^+_{\alpha a} \psi^{\beta b}, \Phi_1, \dots, \Phi_k] = i c'_{\rho,n} D_k(\Phi_1, \dots, \Phi_k)_{\mu I} \kappa^{I J} {T_J}_a^b {\gamma^\mu}_\alpha^\beta,
\]
where $\alpha, \beta$ are spinor indices, $a, b$ gauge indices, $I, J$ Lie algebra indices, $T_J$ the generator in the representation $\rho$, and $\kappa^{I J}$ is the inverse of the Killing form. $c'_{\rho, n}$ is an appropriate normalization factor which depends on the representation $\rho$ and the dimension $n$. Here we used the symbolic notation employed in \eqref{eq:CurrentSuccinct}. By the above, we have fulfilled {\bf Starting element}, {\bf Symmetry}, {\bf Support}, {\bf Scaling}, {\bf Source term}, {\bf Unitarity}, {\bf Microlocal spectrum condition}, {\bf Smoothness}, and {\bf Analyticity}, \cf Remark~\ref{rem:Renormalization}. {\bf Expansion} can be fulfilled by adapting the $r^{k+1}$ for higher grades, \cf \cite{HollandsWaldStress}. The same procedure applies for $G = U(1)$, where $\kappa$ above is replaced by $\kappa(\lambda, \lambda') = \lambda \lambda'$ with $\lambda, \lambda' \in i \R$. Hence, we have shown the following:

\begin{proposition}
Let $n\leq 4$ and $G = G_1 \times \dots \times G_r \times U(1)^l$, where the $G_i$ are simple. Then the time-ordered products can be defined such that \eqref{eq:PPA_gauge} holds.
\end{proposition}

\begin{remark}
For chiral models, this does in general not hold, as it may not be possible to define the parametrix such that \eqref{eq:delta_j} holds, \cf \cite{LocCovDirac}. If, however, the gauge group and representation is such that \eqref{eq:delta_j} can be fulfilled, then also the principle of perturbative agreement can be fulfilled, as at no other place we made any use of the fact that we were dealing with Dirac instead of chiral fields.\footnote{In the proof of Lemma~\ref{lemma:Smoothness}, we used the Dirac equation, but a similar argument should hold for any wave equation with a well-posed Cauchy problem. Furthermore, in the proof of Proposition~\ref{prop:E0}, we excluded the antisymmetric tensor for simplicity, but at least for $n \leq 4$ its inclusion would not change the conclusion.}
\end{remark}

\subsection*{Acknowledgments}

This work was initiated during a stay at the Hausdorff Institute for Mathematics, Bonn, under the program ``Mathematical Physics''. I would like to thank the Hausdorff Institute for hospitality and the program initiators Wojciech Dybalski,  Katarzyna Rejzner, Jan Schlemmer, and Yoh Tanimoto for the kind invitation and for stimulating discussions. Furthermore, discussions with Stefan Hollands, Kartik Prabhu and Harold Steinacker are gratefully acknowledged.

This work was supported by the Austrian Science Fund (FWF) under the contract P24713.

\appendix

\section{A smoothness result}
\label{app:Smoothness}

\begin{proof}[Proof of Lemma~\ref{lemma:Smoothness}]
As $\omega_s$ is a bi-solution, we have
\begin{align*}
 d_s(D^\oplus_s u_1, u_2) & = - H_s(D^\oplus_s u_1, u_2) = G^1_s(u_1, u_2), \\
 d_s(u_1, D^\oplus_s u_2) & = - H_s(u_1, D^\oplus_s u_2) = G^2_s(u_1, u_2), \\
 d_s(D^\oplus_s u_1, D^\oplus_s u_2) & = - H_s(D^\oplus_s u_1, D^\oplus_s u_2) = G^3_s(u_1, u_2).
\end{align*}
As $H_s$ is a bi-solution modulo $C^\infty$, the $G^i_s$ are smooth in $x_1, x_2$. From the Hadamard recursion relations that determine $H$ up to smooth terms, \cf \cite{LocCovDirac}, it follows that they are jointly smooth in $(s, x_1, x_2)$. Furthermore, if $K$ is the support of the perturbation, then $d_s$ is independent of $s$ on any geodesically convex open set which does not intersect $J^+(K)$, by construction. Now choose two Cauchy surfaces $S_\pm$ of $U$ such that $J^\pm(K) \cap S^\mp = \emptyset$. For $u_i \in \D^\oplus$, supported in $N \doteq J^+(S^-) \cap J^-(S^+)$, consider
\[
 d(u_1, u_2) = d(\chi_N u_1, \chi_N u_2) = d(\chi_N D^\oplus S^\oplus_\adv u_1, \chi_N D^\oplus S^\oplus_\adv u_2),
\]
where $\chi_N$ is the characteristic function of $N$ and $d, D^\oplus, S^\oplus_\adv$ depend on $s$. Using integration by parts, we obtain
%\footnote{Note that there is a typo in the corresponding equation (227) in \cite{HollandsWaldStress}.}
\begin{multline}
\label{eq:d_Integral}
 d(u_1, u_2) = G^3(\chi_N S^\oplus_\adv u_1, \chi_N S^\oplus_\adv u_2) + G^1(\chi_N S^\oplus_\adv u_1, \delta_{S^-} n \cdot \gamma S^\oplus_\adv u_2) \\
 + G^2(\delta_{S^-} n \cdot \gamma S^\oplus_\adv u_1, \chi_N S^\oplus_\adv u_2) + d(\delta_{S^-} n \cdot \gamma S^\oplus_\adv u_1, \delta_{S^-} n \cdot \gamma S^\oplus_\adv u_2),
\end{multline}
where $\delta_{S^-}$ denotes the restriction of the integration to $S^-$, and $n$ is the corresponding normal vector.
%In the last term on the r.h.s., the dependence on $s$ is only through $S^\oplus_\adv$, and in particular, 
As $d_s$ is jointly smooth in $(x_1, x_2)$ for fixed $s$ and independent of $s$ in a neighborhood of $S^- \times S^-$, the wave front set of $d$ (as a function of $(s, x_1, x_2)$) has no intersection with $S^- \times S^-$. It now follows from
\[
 \WF(S^{\oplus, s}_\adv) \subset \{ (s, t; x_1, k_1; x_2, k_2) | (x_1, k_1) \sim_s (x_2, - k_2), x_1 \in J^{-}_s(x_2) \}
\]
that $\WF(\chi_N(x_1) S_\adv^\oplus(x_1,x_2))$ does not contain elements with $k_1=0$, so with
the smoothness of $G^i$ and the wave front set calculus \cite{HoermanderI} we conclude that $\WF(d)$ is empty, so $d$ is jointly smooth in $(s, x_1, x_2)$.

In order to restrict the wave front set of $\delta d$, we first determine the wave front set of $\delta G_i$ (which is defined analogously to $\delta d$). Explicit calculation shows that $G_i(x_1, x_2)$ is a series in Hadamard coefficients $V_k(x_1, x_2)$, possibly acted upon with $D^\oplus_1$ or $D^\oplus_2$, and multiplied with nonnegative powers of the squared geodesic distance $\Gamma(x_1, x_2)$. For both $V_k$ and $\Gamma$, one has to integrate the geometric data along the unique geodesic connecting $x_1$ and $x_2$. Hence, for the wave front set of $\delta \Gamma$, $\delta V_k$, we have to determine the wave front set of
\[
 \int_0^1 \delta(y, z_{x_1, x_2}(s)) \ud s,
\]
where $z_{x_1, x_2}: [0,1] \to M$ is the unique geodesic from $x_1$ to $x_2$. To do so, consider the wave front set of $\delta(y, z_{x_1, x_2}(s))$ as a distribution in $(s, y, x_1, x_2)$ and then use the wave front set calculus to determine the wave front set of its convolution with $\vartheta(s) \vartheta(1-s)$, $\vartheta$ being the Heaviside distribution. We obtain, by employing Riemannian normal coordinates,
\begin{align*}
 \WF(\delta V_k) & \subset W, & \WF(\delta \Gamma) & \subset W,
\end{align*}
with
\begin{align*}
 W & \doteq \{ (y,p;x_1,k_1;x_2,k_2) | y = z_{x_1, x_2}(s), 0 \leq s \leq 1, p(\dot z_{x_1, x_2}(s)) = 0, \\
 & \qquad \qquad k_1 = (s-1) \Pi_z p, k_2 = - s \Pi_z p \} \\
 & \cup \{ (y,p;x_1,k_1;x_2,k_2) | y = x_1, p = - k_1, k_2 = 0 \} \\
 & \cup \{ (y,p;x_1,k_1;x_2,k_2) | y = x_2, p = - k_2, k_1 = 0 \},
\end{align*}
where $\Pi_z$ denotes the parallel transport of the cotangent vector along $z$. The wave front set of the $\delta$ distribution $\delta(y-x_i)$ that is implied by $\delta D^\oplus_i V_k(x_1, x_2)$ (as the geometric data at $x_i$ is contained in $D^\oplus_i$) is already contained in $W$ (as the second and the third component), so that\footnote{This corrects equation (230) in \cite{HollandsWaldStress}.}
\begin{equation}
\label{eq:WF_delta_G}
 \WF(\delta G_i) \subset W.
\end{equation}
Furthermore, from $\delta S^\oplus_\adv = - S^\oplus_\adv \circ \delta D^\oplus \circ S^\oplus_\adv$, we have
\begin{multline*}
 \WF(\delta S^\oplus_\adv) \subset \{ (y,p;x_1,k_1;x_2,k_2) | y \in J^+(x_1) \cap J^-(x_2), \\
\exists q_i \in T^*_y M \text{ such that } (y,q_i) \sim (x_i, - k_i), p = q_1 + q_2 \}.
\end{multline*}
Combining this with \eqref{eq:WF_delta_G} and \eqref{eq:d_Integral}, noting that $\delta d=0$ on $S^- \times S^-$, and using the wave front set calculus, we obtain \eqref{eq:WF_delta_d}.
%the bi-distribution
%\[
% \alpha_{\mu \nu}(x_1, x_2) \doteq (S_\adv^{\oplus} u_1)_A(x_1) (S_\adv^{\oplus} u_2)_C(x_2) {\gamma_\mu}^A_B {\gamma_\nu}^C_D d^{BD}(x_1, x_2),
%\]
%where the capital letters denote combined spinor and gauge indices. All objects apart from the $u_i$ and the variables and indices depend on $s$.
\end{proof}

%\bibliography{../mybib}{}
%%\bibliographystyle{amsalpha}
%%\bibliographystyle{../halpha}
%\bibliographystyle{../h-elsevier_new}

\end{document}